\documentclass[jair,twoside,11pt,theapa]{article}
\usepackage{jair, theapa, rawfonts}

\jairheading{-}{-}{-}{-}{-}
\ShortHeadings{Fair Division under Heterogeneous Matroid Constraints}
{Dror, Feldman, Segal-Halevi}
\firstpageno{1}

\usepackage{amsmath,amsthm,amssymb}
\usepackage{booktabs} 
\usepackage[ruled]{algorithm2e} 
\usepackage{algorithmic}

\SetAlFnt{\small}
\SetAlCapFnt{\small}
\SetAlCapNameFnt{\small}
\SetAlCapHSkip{0pt}
\IncMargin{-\parindent}

\usepackage{multirow}
\usepackage{makecell}
\usepackage{colortbl}
\usepackage{hhline}
\usepackage{xcolor}
\definecolor{ForestGreen}{rgb}{.13,.54,.13}

\usepackage{tikz}
\usetikzlibrary{calc}

\usepackage{subcaption}
\usepackage{url}

\theoremstyle{plain}
\newtheorem{theorem}{Theorem}
\newtheorem{lemma}{Lemma}
\newtheorem{assumption}{Assumption}
\newtheorem{observation}{Observation}
\newtheorem{proposition}{Proposition}

\newtheorem{obs}[theorem]{Observation}

\theoremstyle{definition}
\newtheorem{remark}{Remark}
\newtheorem{example}{Example}
\newtheorem{definition}{Definition}

\newcommand{\argmax}{\operatorname*{argmax}}
\newcommand{\argmin}{\operatorname*{argmin}}

\newcommand{\baseord}{base-orderable}
\newcommand{\best}{\textsc{Best}}
\newcommand{\matroid}[1][]{\mathcal{M}_{#1}}
\newcommand{\indSets}[1][]{\mathcal{I}_{#1}}
\newcommand{\feasibleAlloc}{\mathcal{F}}
\newcommand{\G}{\mathcal{G}}
\newcommand{\Category}[1]{C^{#1}}
\newcommand{\category}[2]{\Category{#2}_{#1}}
\newcommand{\capacity}[2]{k_{#1}^{#2}}
\newcommand{\Allocation}[1]{X_{#1}}
\newcommand{\allocation}[2]{\Allocation{#1}^{#2}}
\newcommand{\Valuation}[1]{v_{#1}}
\newcommand{\valuation}[2]{\Valuation{#1}({#2})}
\newcommand{\Fvaluation}[1]{\hat{v}_{#1}}
\newcommand{\fvaluation}[2]{\Fvaluation{#1}({#2})}
\newcommand{\vi}[1][]{\ifthenelse{\equal{#1}{}}{\Valuation{i}}{\valuation{i}{#1}}}
\newcommand{\fvi}[1][]{\ifthenelse{\equal{#1}{}}{\Fvaluation{i}}{\fvaluation{i}{#1}}}
\newcommand{\vone}[1][]{\ifthenelse{\equal{#1}{}}{\Valuation{1}}{\valuation{1}{#1}}}
\newcommand{\fvone}[1][]{\ifthenelse{\equal{#1}{}}{\Fvaluation{1}}{\fvaluation{1}{#1}}}
\newcommand{\vj}[1][]{\ifthenelse{\equal{#1}{}}{\Valuation{j}}{\valuation{j}{#1}}}
\newcommand{\fvj}[1][]{\ifthenelse{\equal{#1}{}}{\Fvaluation{j}}{\fvaluation{j}{#1}}}
\newcommand{\vtwo}[1][]{\ifthenelse{\equal{#1}{}}{\Valuation{2}}{\valuation{2}{#1}}}
\newcommand{\fvtwo}[1][]{\ifthenelse{\equal{#1}{}}{\Fvaluation{2}}{\fvaluation{2}{#1}}}
\newcommand{\vv}[1][]{\ifthenelse{\equal{#1}{}}{\Valuation{}}{\valuation{}{#1}}}
\newcommand{\xx}{\mathbf{\Allocation{}}}
\newcommand{\smartswap}{smart move}
\newcommand{\feasible}[2]{\best_{#1}(#2)}

\newcommand{\fefone}{F-EF1}

\newcommand{\Surplus}[2]{s_{#1}^{#2}}
\newcommand{\surplus}[3]{\Surplus{#1}{#2}(#3)}
\newcommand{\pe}[3][\xx]{\textsc{Envy}^+_{#1}(#2,#3)}

\newcommand{\crr}[4][]{\mathcal{R}({#2},{#3},{#4})_{#1}}
\renewcommand{\S}{{\normalfont Section }}
\newcommand{\newitem}{x^{\text{new}}}

\newcommand{\er}[1]{#1}

\newcommand{\citet}[1]{\citeA{#1}}
\newcommand{\citep}[1]{\cite{#1}}

\begin{document}

\title{On Fair Division under Heterogeneous Matroid Constraints}

\author{\name Amitay Dror 
\email amitaydr@gmail.com
\\
\addr Tel-Aviv University, Tel-Aviv, Israel
\\
\\
\name Michal Feldman 
\email mfeldman@tau.ac.il
\\
\addr Tel-Aviv University, Tel-Aviv, Israel
\\
\AND
\name Erel Segal-Halevi 
\email erelsgl@gmail.com
\\
\addr Ariel University, Ariel 40700, Israel
}

\maketitle
\footnotetext[1]{
A preliminary version appeared in the proceedings of AAAI 2021 \citep{dror2021fair}, without most of the proofs.
This version contains all omitted proofs, 
an uptodate literature survey,
a more general non-existence result in Subsection \ref{sub:matching-ef1},
a simpler proof of Theorem \ref{mnw_identical_vals},
and simpler algorithms and proofs in Section \ref{general_matroids}.
}

\begin{abstract}
We study fair allocation of indivisible goods among additive agents with feasibility constraints. 
In these settings, every agent is restricted to get a bundle among a specified set of feasible bundles.
Such scenarios have been of great interest to the AI community due to their applicability to real-world problems.
Following some impossibility results, we restrict attention to matroid feasibility constraints that capture natural scenarios, such as the allocation of shifts to medical doctors, and the allocation of conference papers to referees.

We focus on the common fairness notion of envy-freeness up to one good (EF1).
Previous algorithms for finding EF1 allocations are either restricted to agents with identical feasibility constraints, or allow free disposal of items.
An open problem is the existence of EF1 complete allocations among heterogeneous agents, where the heterogeneity is both in the agents' feasibility constraints and in their valuations.
In this work, we make progress on this problem by providing positive and negative results for different matroid and valuation types.
Among other results, we devise polynomial-time algorithms for finding EF1 allocations in the following settings: (i) $n$ agents with heterogeneous partition matroids and heterogeneous binary valuations, (ii) 2 agents with heterogeneous partition matroids and heterogeneous additive valuations, and (iii) at most 3 agents with heterogeneous binary valuations and identical base-orderable matroid constraints.
\end{abstract}

\section{Introduction}
\label{introduction}

Many real-life problems involve the fair allocation of indivisible 
items among agents with different preferences, and with constraints on the bundle that each agent may receive.
Examples include the allocation of course seats among students  
\citep{budish2017course} and the allocation of conference papers among referees
\citep{garg2010assigning}.

In general, different agents may have different constraints.
For example, consider the allocation of employees among departments of a company: one department has room for four project managers and two backend engineers, while another department may have room for three backend engineers and five data scientists.
Another example can be found in the way shifts are assigned among medical doctors, where every doctor has her own schedule limitations.   

Our goal is to devise algorithms that find fair allocations of indivisible items among agents with different preferences and different feasibility constraints.
Let us first explain what we mean by ``fair'' and what we mean by ``constraints''.

A classic notion of fairness is {\em envy freeness} (EF), which means that every agent (weakly) prefers  his or her bundle to that of any other agent. 
Since an EF allocation may not exist when items are indivisible, 
recent studies focus on its relaxation known as {\em EF1} --- envy free up to one item \citep{budish2011combinatorial} --- which means that every agent $i$ (weakly) prefers her bundle to any other agent $j$’s bundle, up to the removal of the best good (in $i$'s eyes) from agent $j$'s bundle (Section  \ref{Fairness_Notions}).
Without constraints, an EF1 allocation always exists and can be computed efficiently \citep{lipton2004approximately}. 

The constraints of an agent are represented by a set of bundles (subsets of items), that are considered \emph{feasible} for the agent. An allocation is feasible if it allocates to each agent a feasible bundle. We focus on the case when the feasible bundles are the \emph{independent sets of a matroid}. This means that (i) the set of feasible bundles is \emph{downward-closed} --- 
a subset of a feasible bundle is feasible; (ii) if a feasible bundle $S$ has fewer items than another feasible bundle $T$, then it is possible to extend $S$ to a larger feasible bundle by adding some item from $T$%
. This latter property of ``extension by one'' makes the notion of EF1 particularly appropriate for problems of allocation with matroid constraints. 
A special case of a matroid is a {\em partition matroid}.
With partition matroid constraints, the items are partitioned into a set of {\em categories}, each category has a \emph{capacity}, and the feasible bundles are the bundles in which the number of items from each category is at most the category capacity.

There are two approaches for handling feasibility constraints in fair allocation.
A first approach is to directly construct allocations that satisfy the constraints, i.e., guarantee that each agent receives a feasible bundle. 
This approach was taken recently by \citet{biswas2018fair,biswas2019matroid}, who study settings with \emph{additive valuations}, where every agent values each bundle at the sum of the values of its items.
They present efficient algorithms for computing EF1 allocations when agents have: (i) identical matroid constraints and identical valuations; or (ii) identical partition matroid constraints, even under heterogeneous valuations (Section  \ref{sub:common-tools}).
However, their algorithms do not handle different partition constraints, or identical matroid constraints with different valuations. 

A second approach is to capture the constraints within the valuation function. That is, the value of an agent for a bundle equals the value of the best feasible subset of that bundle.
This approach seamlessly addresses heterogeneity in both constraints and valuations.
The valuation functions constructed this way are no longer additive,
but are \emph{submodular} [\cite{oxley2006matroid}].
Recently, \citet{babaioff2020fair} and \citet{benabbou2020finding} have independently proved the existence of EF1 allocations in the special case in which agents have submodular valuations with \emph{binary marginals} (where adding an item to a bundle adds either $0$ or $1$ to its value). 
Such an allocation can be converted to a fair \emph{and feasible} allocation by giving each agent the best feasible subset of his/her allocated bundle, and disposing of the other items.

However, in some settings, such disposal of items may be impossible.
For example, when allocating shifts to medical doctors, 
if an allocation rule returns an infeasible allocation and shifts are disposed to make it feasible, the emergency-room might remain understaffed. 
A similar problem may occur when allocating papers to referees, where disposals may leave some papers without reviews.
The allocation rules developed in the above papers may not yield EF1 allocations when they are constrained to return feasible allocations (Section \ref{sub:partition-mnw}).
Thus, an open problem remains:

\vspace{0.05in} 
\noindent \emph{\textbf{Open problem}: Given agents with \emph{different} additive valuations and \emph{different} matroid constraints, which settings admit a \emph{complete} and \emph{feasible} EF1 allocation?}

\subsection{Contribution and Techniques}

\vspace{0.05in}
\noindent {\bf Feasible envy.}
Before presenting our results, we shall discuss the EF1 notion in settings with heterogeneous constraints.
Consider a setting with two agents, Alice and Bob, and 8 identical items of a single category, with capacities 3 and 5 for Alice and Bob, respectively. Every complete feasible allocation gives 3 items to Alice and 5 to Bob. 
Ignoring feasibility constraints, such an allocation is not EF1, since even after removing a single item from Bob's bundle, Alice values it at 4, which is greater than her value for her own bundle. 
However, a bundle of 4 items is infeasible for Alice, so her envy is not justified.

A more reasonable definition of envy in this setting is 
\emph{feasible envy}, where each agent compares her bundle to the best feasible subset of any other agent's bundle.
In the example above, the best feasible subset of Bob's bundle for Alice is worth 3. Thus, the allocation is \emph{feasibly-envy-free} (F-EF).

If Alice values one of Bob's items at 2, then the above allocation is not F-EF, since the best feasible subset of Bob's bundle for Alice is worth 4, but it is {\em \fefone}, as it becomes F-EF after removing this item from Bob's bundle (Section  \ref{Fairness_Notions} for a formal definition). Throughout the paper, we use the notion of F-EF1 under heterogeneous constraints. Note that \fefone{} is equivalent to EF1 when agents have identical constraints.

\vspace{0.05in}
\noindent {\bf Impossibilities.}
We present several impossibility results that direct us to the interesting domain of study.
First, if the partition of items into categories is different for different agents, an F-EF1 allocation may not exist, even for two agents with identical valuations (Section ~\ref{sub:different_categories}).
Second, going beyond matroid constraints to natural generalizations such as \emph{matroid intersection}, \emph{bipartite graph matching}, \emph{conflict graph} or \emph{budget constraints} is futile: even with two agents with identical binary valuations and identical non-matroid constraints, a complete F-EF1 allocation may not exist (Section ~\ref{sub:matching-ef1}).
Third, going beyond EF1 to the stronger notion of \emph{envy-free up to any good (EFX)} is also hopeless:  
even with two agents with identical valuations and identical uniform matroid constraints, an EFX allocation may not exist (Section  \ref{sub:uniform-efx}).

Based on these results, 
we focus on finding F-EF1 allocations
when the agents' constraints are represented by either:
(1) \emph{partition matroids} where all agents share the same partition of items into categories but may have different capacities;
or
(2) \emph{\baseord{} (BO) matroids} --- a wide class of matroids containing partition matroids ---  where all agents have identical matroid constraints but possibly different valuations.

\vspace{0.05in}
\noindent {\bf Algorithms (see Table~\ref{tab:results}).}
For {\em partition matroids}, the reason that the algorithms of \citet{lipton2004approximately} and \citet{biswas2018fair} fail for agents with different capacities is that they rely on \emph{cycle removal} in the envy graph. Informally (Section  \ref{sub:common-tools} for details), these algorithms
maintain a directed \emph{envy graph} in which each agent points to every agent he or she envies. The algorithm prioritizes the agents who are not envied, since giving an item to such agents keeps the allocation EF1. If there are no unenvied agents, the envy graph must contain a cycle, which is then removed by exchanging bundles along the cycle. 
However, when different agents in the cycle have different constraints, this exchange may not be feasible.
Thus, our main challenge is to develop techniques that guarantee that no envy-cycles are created in the first place. We manage to do so in four settings of interest:
\begin{enumerate}
\item There are at most \emph{two categories} 
(Section  \ref{partition_warmup}).
\item All agents have \emph{identical} valuations (Section  \ref{sec:identical-valuations}).
\item All agents have \emph{binary valuations}  (Section  \ref{partition-binary}).
\item There are \emph{two agents} (Section  \ref{2_agents_sec}).
\end{enumerate}
Each setting is addressed by a different algorithm and using a different cycle-prevention technique.

Beyond partition matroids, we consider the much wider class of matroids, termed {\em base-orderable} (BO) matroids (see definition~\ref{def:bo}). This class contains partition matroids, laminar matroids (an extension of partition matroids
where the items in each category can be partitioned into sub-categories), transversal matroids, and other interesting matroid structures. In fact, it is conjectured by \citet{bonin2016infinite} that ``almost all matroids are base-orderable".
For this class we present algorithms for agents with identical constraints and different additive valuations in the following cases:
\begin{enumerate}
\setcounter{enumi}{4}
\item There are \emph{two agents} (Section \ref{general_matroids}).
\item There are \emph{three agents} with \emph{binary valuations} (Section \ref{general_matroids}).
\end{enumerate}
All our algorithms run in polynomial-time.

\newcommand{\gry}{\cellcolor[gray]{0.95}}
\newcommand{\parthline}{\cline{2-9}}
\begin{table*}
\centering
\footnotesize
 \begin{tabular}{|c|c|c|c|c|c|c|c|c|} 

\hline
\makecell{\footnotesize Matroid \\\footnotesize Type} &  \makecell{\footnotesize Complete\\ \footnotesize allocation}  & \makecell{\footnotesize Het.\\ \footnotesize constraints} & \makecell{\footnotesize Het.\\ \footnotesize valuations} &
PE &
 \makecell{\footnotesize Valuations} 
& \makecell{\footnotesize \# of \\\footnotesize Agents} & \footnotesize Remarks &\footnotesize  Reference \\
\hline\hline


\multirow{5}{*}{\rotatebox[origin=c]{90}{Partition}}  
& \gry \checkmark  & \gry - & \gry \checkmark & \gry - & \gry General & \gry $n$ & \gry 
& \gry  \footnotesize  \color{ForestGreen} B\&B(2018) 
\\
\parthline
& \checkmark & \checkmark & \checkmark & - & General & $n$ & \makecell{\footnotesize $\leq 2$\\ \footnotesize categories} 
& \makecell{
\color{ForestGreen}
Section \ref{sec:2-categories}
} 
\\
\parthline
& \checkmark & \checkmark & - & \checkmark  & General & $n$ 
& 
& 
\color{ForestGreen}
Section \ref{sec:identical-valuations}
\\
\parthline
& \checkmark & \checkmark & \checkmark & \makecell{\checkmark Binary \\caps} & Binary & $n$ &
& 
\color{ForestGreen}
\S \ref{partition-binary} 
\\
\parthline
& \checkmark & \checkmark & \checkmark & - &General & $2$ & & 
\color{ForestGreen}  
\S \ref{2_agents_sec} 
\\
\parthline
& \checkmark & 
\makecell{\checkmark,   Het.  \\ categories }& 
\makecell{Even\\identical}
& -
& \makecell{Even\\binary} & $2$ &
& 
   \makecell{
   \color{red}
   No F-EF1;\\ 
   \color{red}
   Ex.
   \ref{exp:different_categories} 
   }
   \\
\parthline
& \checkmark & 
\makecell{Even\\identical}
& 
\makecell{Even\\identical}
& -
  & \makecell{Even\\binary} & $2$ &
 & 
   \makecell{
   \color{red}
   No EFX;\\ 
   \color{red}
   Ex.
   \ref{exp:efx} 
   }
   \\   
\hline
\hline

\multirow{4}{*}{\rotatebox[origin=c]{90}{Beyond Partition    }} 
 & \gry \checkmark & \gry - & \gry - & \gry - 
 & \gry General & \gry $n$ & \gry \makecell{Laminar\\matroid} &
 \color{ForestGreen} 
 \gry \footnotesize B\&B(2018)
 \\
\hhline{|~*7{>{\arrayrulecolor{black}}-}|}

 & \gry  - & \gry \checkmark & \gry \checkmark 
& \gry \checkmark  & \gry Binary & \gry $n$ &
\multicolumn{1}{|>{\columncolor[gray]{0.95}}c|}{\makecell{General\\matroid}}
&
 \gry 
 \makecell{
 \color{ForestGreen}
  *** }  
  \\
\parthline

 &  \checkmark & - & \checkmark & \checkmark & Binary & 3 & \makecell{BO Matroid} &
\color{ForestGreen}
Section  \ref{sub:matroid-3agents-binary}
\\
\parthline
 &  \checkmark & - & \checkmark & - & General & 2 & \makecell{BO Matroid} &  
\color{ForestGreen} 
Section  \ref{sub:matroid-2agents-additive} 
\\
\parthline
&  \checkmark & 
\makecell{Even\\identical}
&
\makecell{Even\\identical}
& 
- &
\makecell{Even\\binary} & $2$ &
\makecell{\footnotesize Matroid \\ \footnotesize intersection;\\ \footnotesize Conflict graph;\\ \footnotesize Budget}
 &  
\makecell{
\color{red}
No F-EF1;\\ 
\color{red}
Ex.
\ref{exp:intersection},
\\
\color{red}
\ref{exp:conflict},
\ref{exp:budget}.
}
\\
 \hline
\end{tabular}
\caption{
\label{tab:results}
A summary of our results in the context of previous results. All results are for additive valuations. Gray lines represent previous results. PE refers to outcomes that are also Pareto-efficient. BO refers to \baseord{} matroids (Section  \ref{general_matroids}).
B\&B(2018) is \citet{biswas2018fair}.
The row marked by *** 
follows from 
\citet{benabbou2020finding,babaioff2020fair} (Section  1).
Green text denotes a positive result;
red text denotes an impossibility result.
}
\end{table*}

\subsection{Related Work}
\label{related_work}

A recent survey of constraints in fair division is given by 
\citet{suksompong2021constraints}.
Below we focus on constraints in allocation of indivisible items.

\paragraph{Capacity constraints.}
In many settings, there are lower bounds as well as upper bounds on the total number of items allocated to an agent. This is particularly relevant to the problem of \emph{assigning conference papers to referees} \citep{garg2010assigning,long2013good,lian2017conference}. 
The constraints may be different for each agent, but there is only one category of items. 
The same is true in the setting studied by
\citet{ferraioli2014regular}, where each agent must receive exactly $k$ items.
A \emph{balanced allocation} is an allocation in which all agents receive the same number of items up to at most a single item. The round-robin algorithm finds a balanced EF1 allocation for any number of agents with additive valuations. 
An algorithm by \citet{kyropoulou2020almost} finds a balanced EF1 allocation for two agents with general monotone valuations. It is open whether this result extends to three or more agents.
\citet{jojic2021splitting} prove the existence of \er{a balanced allocation in which all agents assign approximately the same value to all bundles (an ``almost consensus'' allocation).}
Note that 
\er{constraints imposing a lower bound on the number of items allocated to an agent, such as balancedness constraints, are not  matroid constraints, since they are not downward-closed. }

Capacity constraints are common also in matching markets such as doctors--hospitals and workers--firms; see \citet{klaus2016matching} for a recent survey. In these settings, the preferences are usually represented by ordinal rankings rather than by utility functions, and the common design goals are Pareto efficiency, stability and strategy-proofness rather than fairness.

\citet{gafni2021unified} study fair allocation in a related setting in which items may have multiple copies.

\paragraph{Partition constraints.}
Fair allocation of items of different categories is studied by
\citet{mackin2016allocating} and \citet{sikdar2017mechanism}. Each category contains $n$ items, and each agent must receive exactly one item of each category.
\citet{sikdar2019mechanism} consider an exchange market in which each agent holds multiple items of each category and should receive a bundle with exactly the same number of items of each category.
The above  works focus on designing strategyproof mechanisms.
\citet{nyman2020fair} study a similar setting (they call the categories ``houses'' and the objects ``rooms''), but with monetary transfers (which they call ``rent'').

Following the paper by \citet{biswas2018fair} focusing on EF1 fairness, 
\citet{hummel2021guaranteeing} study partition matroid constraints in combination with a different fairness notion --- the \emph{maximin-share}. Their algorithm attains a $1/2$-factor approximation to this fairness notion.

\paragraph{Matroid constraints.}
\citet{gourves2013protocol} study a setting with a single matroid, where the goal includes building a base of the matroid and providing worst case guarantees on the agents' utilities. 
\citet{gourves2014near} and \citet{gourves2019maximin} require the \emph{union} of bundles allocated to all agents to be an independent set of the matroid. This by design requires to leave some items unallocated, which is not allowed in our setting.

\paragraph{Budget constraints.}
\emph{Budget constraints} (also called \emph{knapsack constraints}) assume that each item has a cost, each agent has a budget, and the feasible bundles for an agent are the bundles with total cost at most the budget. 
In this setting, 
\citet{wu2021budget} show that a 1/4-factor EF1 allocation exists.
\citet{gan2021approximately} show that a 1/2-factor EF1 allocation exists when the valuations are identical,
and an EF1 allocation exists among two agents with the same budget.

\paragraph{Connectivity constraints.}
\citet{barrera2015discrete}, \citet{bilo2018almost}, and \citet{suksompong2019fairly}
study another kind of constraint in fair allocation. 
The goods are arranged on a line, and each agent must receive a connected subset of the line.
\citet{bouveret2017fair} and \citet{bei2019connectivity} study a more general setting in which the goods are arranged on a general graph, and each agent must receive a connected subgraph. 
Note that these are not matroid constraints.

\paragraph{No-conflict constraints.}
\citet{lili2021fair} study fair allocation with \emph{scheduling constraints}, where each item is an interval, and the feasible bundles are the sets of non-overlapping intervals.
\citet{hummel2021fair} study fair allocation in a more general setting in which there is a \emph{conflict graph} $G$, and the feasible sets are the sets of non-adjacent vertices in $G$. 

\paragraph{Downward-closed constraints.}
\citet{livetta2021fair} study fair allocation with  \emph{downward-closed constraints}, which include matroid, budget, and no-conflict constraints as special cases. For this very general setting, they present an algorithm that approximates the maximin share.

\paragraph{Non-additive valuations.}
As explained in the introduction, fair allocation with constraints is closely related (though not equivalent) to \emph{fair allocation with non-additive valuations}. 
This problem has attracted considerable attention recently. 
\citet{bei2017earning} and \citet{anari2018nash} study allocation of multi-unit item-types when the valuation is additive between types but concave (i.e., has decreasing marginal returns) between units of the same type. They give a 2-approximation to the maximum Nash welfare (the product of utilities).
\citet{garg2018approximating} study budget-additive valuations, where each agent has a utility-cap and values each bundle as the minimum between the sum of item values and the utility-cap. They give a 2.404-approximation to the MNW.

Particularly relevant to matroid constraints are the 
\emph{submodular valuations with binary marginals}, where adding an item to a bundle increases the bundle value by either $0$ or $1$. 
These valuations are equivalent to \emph{matroid rank functions} --- functions that evaluate a bundle by the size of the largest independent set of a certain matroid contained in that bundle. 
In this setting, \citet{benabbou2020finding} and \citet{babaioff2020fair} present allocation mechanisms that are Pareto-efficient, EF1, EFX and strategyproof.  \citet{barman2020existence} present a polynomial-time algorithm that finds an allocation satisfying maximin-share fairness.

\section{Model and Preliminaries} 
\label{model_and_preliminaries}

\subsection{Allocations and Constraints}
\label{sub:allocations}
We consider settings where a set $M$ of $m$ items should be allocated among a set $N$ of $n$ agents. 
An allocation is denoted by $\xx=(\Allocation{1},\dots,\Allocation{n})$, where $\Allocation{i} \subseteq M$ is the bundle given to agent $i$, and $\Allocation{i}\cap \Allocation{j} = \emptyset$ for all $i\neq j \in N$. 
An allocation is  \emph{complete} 
if $\biguplus_{i\in N}{\Allocation{i}} = M$.
Throughout, we use $[n]$ to denote the set $\{1,\ldots,n\}$.

We consider constrained settings, where every agent $i$ is associated with a matroid $\matroid[i]=(M,\indSets[i])$ that specifies the feasible bundles for $i$. 

\begin{definition}
\label{def:matroid}
A {\em matroid} is a pair $\matroid=(M,\indSets)$,
where $M$ is a set of items and $\indSets \subseteq 2^M$ is a nonempty set of {\em independent sets} satisfying the following  properties: 

(i) {\em Downward-closed}: 
$S \subset T$ and $T \in \indSets$ implies
$S \in \indSets$;

(ii) {\em Augmentation}: For every $S,T\in \indSets$, if $|S|<|T|$, then $S\cup\{g\}\in \indSets$ for some $g\in T$.

A \emph{base} of  $\matroid$ is a maximal independent set in $\matroid$.
\end{definition}

A special case of a matroid is a {\em partition matroid}:
\begin{definition} (partition matroid)
	A matroid $\matroid[i] = (M,\indSets[i])$ is a {\em partition matroid} if
	there exists a partition of $M$ into \emph{categories} $C_i=\{\category{i}{1},\ldots,\category{i}{\ell_i}\}$ for some $\ell_i \leq m$, 
	and a corresponding vector of {\em capacities} $\capacity{i}{1},\ldots,\capacity{i}{\ell_i}$,
	such that the collection of independent sets is
$$
\indSets[i] = \{ S\subseteq M : |S\cap \category{i}{h}| \leq \capacity{i}{h} \text{ for every } h\in[\ell_i]\}.
$$
\end{definition}

Given an allocation $\xx$, we denote by $\xx_i^{h}$ the items from category $\category{i}{h}$ given to agent $i$ in $\xx$.

A special case of a partition matroid is a {\em uniform matroid}, which is a partition matroid with a single category.



\vspace{0.05in}
\begin{definition}(feasible allocation)
	An allocation $\xx{}$ is said to be {\em feasible} if:
	(i) it is individually feasible: $\Allocation{i} \in \indSets[i]$ for every agent $i$, and
	(ii) it is complete: $\biguplus_i{\Allocation{i}}=M$.
\end{definition}
Let $\feasibleAlloc$ denote the set of all feasible allocations. 
Throughout this paper we consider only instances that admit a feasible allocation:
\begin{assumption}
\label{asm:feasible}
All instances considered in this paper admit a feasible allocation; i.e., $\feasibleAlloc\neq\emptyset$.
For partition matroids, feasibility means that for every category $\Category{h}$, the sum of agent capacities for this category is at least $|\Category{h}|$.
\end{assumption}

An instance is said to have {\em identical matroids} if all agents have the same matroid feasibility constraints. I.e., $\indSets[i] = \indSets[j]$ for all $i,j \in N$. 

An instance with partition matroids is said to have {\em identical categories} if all the agents have the same partition into categories. I.e., $\ell_i = \ell_j = \ell$ for every $i,j \in N$, and $\category{i}{h}=\category{j}{h}=\Category{h}$ for every $h \in \ell$. The capacities, however, may be different. 

\subsection{Valuations and Fairness Notions} 
\label{Fairness_Notions}

Every agent $i$ is associated with an {\em additive} valuation function $v_i:2^M \rightarrow R^+$, which assigns a positive real value to every set $S \subseteq M$. Additivity means that there exist $m$ values $v_i(1),\ldots,v_i(m)$ such that $v_i(S)= \sum_{j\in S}v_i(j)$. 
An additive valuation $\Valuation{i}$ is called {\em binary} if $\valuation{i}{j}\in\{0,1\}$ for every $i \in N ,j \in M$. An allocation $\xx$ is \emph{Social Welfare Maximazing} (SWM) if $\xx \in \argmax_{\xx'\in\feasibleAlloc}{\sum_{i\in[n]}\vi[\Allocation{i}']}$.

\begin{definition}[envy and envy freeness]
Given an allocation $\xx{}$, agent $i$ $\emph{envies}$ agent $j$ if $\valuation{i}{\Allocation{i}}<\valuation{i}{\Allocation{j}}$. 
$\xx{}$ is {\em envy free} if no agent envies another agent. 
\end{definition}

\begin{definition}[EF1]\citep{budish2011combinatorial}
An allocation $\xx{}$ is \emph{envy free up to one good (EF1)} if for every $i,j \in N$, 
there exists a subset $Y \subseteq  \Allocation{j}$ with $|Y|\leq 1$, such that $\valuation{i}{\Allocation{i}} \geq \valuation{i}{\Allocation{j}\setminus Y}$.
\end{definition}

\begin{definition}
A {\em best feasible subset} of a set $T$ for agent $i$, denoted $\feasible{i}{T}$ is any subset in 
$$
\argmax_{S\subseteq T, \ S\in \indSets[i]} \valuation{i}{S}.
$$
\end{definition}


\begin{definition}[feasible valuation]
The {\em feasible valuation} of agent $i$ for a set $T$ is
$
\fvaluation{i}{T} := \valuation{i}{\feasible{i}{T}}
$.
\end{definition}
Note that $\fvaluation{i}{T}$ is well-defined even though $\feasible{i}{T}$ may not be uniquely determined.

\begin{definition}
Given a feasible allocation $\xx{}$:
\begin{itemize}
	\item Agent $i$ \emph{F-envies} agent $j$ iff $\fvaluation{i}{\Allocation{i}}<\fvaluation{i}{\Allocation{j}}$.
	\item $\xx{}$ is \emph{F-EF} (feasible-EF) if no agent F-envies another one.
	\item $\xx{}$ is \emph{\fefone} if for every $i,j \in N$:
there exists a subset $Y \subseteq  \Allocation{j}$ with $|Y|\leq 1$, such that $\fvaluation{i}{\Allocation{i}} \geq \fvaluation{i}{\Allocation{j}\setminus Y}$.
\end{itemize}
\end{definition}
For further discussion of the \fefone{} criterion, and an alternative (weaker) definition, see Appendix~\ref{sub:fefone}.
~~
Another useful notation is \emph{positive feasible envy}, which is the amount by which an agent F-envies another agent:
\begin{definition}
The \emph{positive feasible envy} of agent $i$ towards $j$ in allocation $\xx$ is:
$$
\pe{i}{j} := \max(0, \fvi[{\Allocation{j}}]-\fvi[{\Allocation{i}}]).
$$
\end{definition}

\begin{definition}
The \emph{envy graph} of an allocation $\xx$, $\G(\xx)$, is a directed graph where the nodes represent the agents, and there is an edge from agent $i$ to agent $j$ iff $\vi[\Allocation{i}] < \vi[\Allocation{j}]$.
The \emph{feasible envy graph} is defined analogously based on the feasible-envy. 
\end{definition}




\subsection{Common Tools and Techniques}
\label{sub:common-tools}
Below we review the most common methods for finding an EF1 allocation.

\paragraph{Envy cycle elimination.}
The first method for attaining an EF1 allocation (in unconstrained setting, even with arbitrary monotone valuations) is due to \citet{lipton2004approximately}.

The \emph{envy cycles elimination} algorithm works as follows. Start with the empty allocation. 
Then, as long as there is an unallocated item: (i) choose an agent that is a source in the envy graph (not envied by another agent) and give her an arbitrary unallocated item, (ii) reconstruct the envy graph $\G$ corresponding to the new allocation, (iii) as long as $\G$ contains cycles, choose an arbitrary cycle, and shift the bundles along the cycle. This increases the total value, thus this process must end with a cycle-free graph.


\paragraph{Max Nash welfare.}
The \emph{Nash social welfare} (NW) of an allocation $\xx$ is the geometric mean of the agents' values: $NW=(\prod_{i\in[n]}\vi[\Allocation{i}])^\frac{1}{n}$.
An allocation is {\em max Nash welfare (MNW)} if it maximizes the NW among all feasible allocations.
\citet{caragiannis2019unreasonable} showed that in unconstrained settings with additive valuations, every MNW allocation is EF1.

\paragraph{Round robin (RR).}
RR works as follows. Given a fixed order $\sigma$ over the agents, as long as there is an unallocated item, the next agent according to $\sigma$ (where the next agent of agent $n$ is agent $1$) chooses an item she values most among the unallocated items.
Simple as it might be, this algorithm results in an EF1 allocation in unconstrained settings with additive valuations \citep{caragiannis2019unreasonable}

\paragraph{Per category RR + envy cycle elimination.}
This algorithm (Algorithm \ref{alg:pc-rr}) was introduced by \citet{biswas2018fair} for finding an EF1 allocation in settings with homogeneous partition constraints. It resolves the categories sequentially, resolving each one by RR followed by envy cycle elimination, where the order over the agents is determined by a topological order in the obtained envy graph. 

\begin{algorithm}[h]
	\SetAlgoLined
	\textbf{initialize:} \\
	$\sigma \leftarrow$ an arbitrary order over the agents.\\
	$\forall i\in[n] \ \Allocation{i}\leftarrow \emptyset$\\
	\For{every category $h$}{
		Run round robin with $\Category{h}, \sigma$;\\
		Let $ \allocation{i}{h}$ be the resulting allocation for agent $i$;\\
		$\forall i\in[n] \ \Allocation{i}\leftarrow \Allocation{i}\cup \allocation{i}{h}$;\\
		Draw envy graph for current allocation;\\
		Remove cycles from the graph, switching bundles along the cycles;\\
		Set $\sigma$ to be a topological order of the graph;
	}
	\caption{Per-Category Round Robin, \cite{biswas2018fair}
		\label{alg:pc-rr}
	}
\end{algorithm}

\subsection{Pareto Efficiency}

\begin{definition}
An allocation $\xx$ is \emph{Pareto-efficient} if there is no other allocation $\xx'$ such that all agents weakly prefer $\xx'$ over $\xx$ and at least one agent strictly prefers $\xx'$.
\end{definition}

We observe that, with additive identical valuations,
every complete feasible allocation is Pareto-efficient.

\begin{obs}
\label{obs:pareto}
For any (possibly constrained) setting with identical additive valuations, 
every complete feasible 
allocation is Pareto efficient
.
\end{obs}
\begin{proof}
Let $\vv$ denote the common valuation of all agents.
Let $\xx$ be a complete feasible allocation.
Feasibility implies that for every agent $i\in N$: $\fvi(\Allocation{i})=\vv[\Allocation{i}]$.
Therefore, 
$\sum_{i\in N} \fvi(\Allocation{i})=\sum_{i\in N}\vv[\Allocation{i}]$.
By completeness and additivity, the latter sum equals $\vv(M)$, which is a constant independent of $\xx$. 
So the sum of feasible values is the same in all complete feasible allocations. 
Therefore, if any other allocation gives a higher value to some agent, it must give a lower value to some other agent. 
\end{proof}

\section{Impossibility Results}
\label{sec:impossibility}

In this section we give some intuition for why previous approaches fail in the case of heterogeneous constraints, and provide impossibility results for settings beyond the ones considered in this paper.

All examples in this section involve two agents. 
Every item is denoted by a pair of values, where the first and second values correspond to the value of the first and second agents, respectively. For example, an item $(0,1)$, or simply $0,1$, denotes an item that agent 1 values at 0 and agent 2 values at 1.

\subsection{Partition Matroids, Maximum Nash welfare}
\label{sub:partition-mnw}
The following example shows that a maximum Nash welfare (MNW) outcome may not be \fefone\ in settings with  feasibility constraints, even under identical partition matroid constraints and binary valuations. We note that the existence of such an example has been mentioned by \citet{biswas2018fair} without a proof; we include it here for completeness.
\begin{example}
\label{exp:partition-mnw}
Consider the setting and allocation illustrated in Table~\ref{tab:mnw_non_ef1}. 
This example consists of 2 agents, and 2 categories. 
Category 1 has 4 items and capacity 2 for both agents; 
Category 2 has 6 items and capacity 3 for both agents. 
Recall that $x,y$ refers to an item that is valued at $x$ by agent $A$ and valued at $y$ by agent $B$.
In the allocation given in Table \ref{tab:mnw_non_ef1}, $\valuation{A}{\Allocation{A}}=2$ and $\valuation{B}{\Allocation{B}}=3$, resulting in Nash welfare of 6. 
One can verify that this allocation is not EF1. 
The only other feasible allocations have either value $0$ for agent $A$ and value $5$ for agent $B$ (for Nash welfare of 0), or value 1 for agent $A$ and value 4 for agent $B$ (for Nash welfare of 4); the latter allocation is EF1.
\end{example}

\begin{table}[h]
\centering
\begin{tabular}{|c||c|c|}
\hline
Category & Alice & Bob\\
\hline\hline
$\Category{1}$ & 1,1 & 0,0\\
 $\capacity{A}{1}=\capacity{B}{1}=2$& 1,1 & 0,0\\
 \hline
 $\Category{2}$ & 0,1 & 0,1\\
 $\capacity{A}{2}=\capacity{B}{2}=3$ & 0,1 & 0,1\\
  & 0,1 & 0,1\\
\hline
\end{tabular}
\caption{
\small
\label{tab:mnw_non_ef1}
An example of agents with identical partition matroid constraints and binary valuations where MNW does not imply EF1.
}
\end{table}

Note that \citet{benabbou2020finding} prove that MNW always implies EF1 for submodular valuations with binary marginals.  However, they consider \emph{clean} allocations, where items with $0$ marginal value are not allocated. In contrast, we consider \emph{complete} allocations, in which all items must be allocated.
Indeed, if we ``clean'' the allocation in Table \ref{tab:mnw_non_ef1} by having Alice dispose of the three items she does not desire in category $\category{}{2}$, the allocation becomes EF1 while remaining MNW.
The same reasoning (and same example) applies also to the \emph{prioritized egalitarian} mechanism introduced by \citet{babaioff2020fair}. 
Specifically, this mechanism gives a {\em clean} (``non-redundant" in their terminology) Lorentz-dominating allocation, which is shown to be EF1 (and even EFX). However, the obtained allocation is not complete. 

\subsection{Partition Matroids, Heterogeneous Categories}
\label{sub:different_categories}

The following example shows that if agents have partition matroid constraints, where the partitions of the items into categories is not the same, then an \fefone{} allocation might not exist, even if the valuations are identical and binary.

\begin{example}[Heterogeneous categories, no \fefone]
\label{exp:different_categories}
Consider a setting 
with four items and two agents with identical binary valuations: $(1,1)$, $(1,1)$, $(0,0)$, $(0,0)$.
Suppose that
\begin{itemize}
\item Alice's partition has two categories, 
$\category{A}{1}=\{(1,1),(0,0)\}$ with capacity $1$ 
and 
$\category{A}{2}=\{(1,1),(0,0)\}$ with capacity 1.
\item 
Bob's partition has three categories: $\category{B}{1}=\{(1,1)\}$ with capacity $1$, 
$\category{B}{2}=\{(1,1)\}$ with capacity $1$, 
and $\category{B}{3}=\{(0,0),(0,0)\}$ with capacity $0$.
\end{itemize}

There is a unique feasible allocation, in which Bob gets the two $(1,1)$ items from $\category{B}{1}$ and $\category{B}{2}$,
and Alice gets the two $(0,0)$ items from 
$\category{A}{1}$ and $\category{A}{2}$.
Bob's bundle is feasible for Alice, so  Alice envies Bob beyond F-EF1.
\end{example}

\subsection{Non-Existence of EF1 for Non-Matroid Constraints}
\label{sub:matching-ef1}
In this subsection we consider natural classes of constraints (set systems) that are not matroidal, and show that they  might not admit a complete EF1 allocation.
\begin{definition}
Let $\indSets\subseteq 2^M$ be a set of subsets of $M$.
Two elements $x,y\in M$ are called \emph{complementary for $\indSets$} 
if, for every partition of $M$ into two subsets $X_1,X_2\in \indSets$ (with $X_1\cap X_2 = \emptyset$ and $X_1\cup X_2 = M$), 
either $\{x,y\}\subseteq X_1$
or $\{x,y\}\subseteq X_2$.
\end{definition}

\begin{proposition}
\label{prop:complementary}
Suppose there are two agents with identical constraints represented by some $\indSets\subseteq 2^M$. If there are complementary items for $\indSets$, then a complete feasible  EF1 allocation might not exist even when the agents have identical binary valuations.
\end{proposition}
\begin{proof}
Suppose both agents value both complementary items at $1$ and the other items at $0$. In every feasible allocation, one agent gets none of the complementary items, and thus envies the other agent who gets both these items, and the envy is by two items.
\end{proof}

There are several natural constraints with complementary items.
\begin{example}[Matching and matroid-intersection constraints]
\label{exp:intersection}
Consider the allocation of course seats among students, where the seats are categorized both by their time and by their subject, and a student should get at most a single seat per time and at most a single seat per subject.
Suppose there are two subjects, physics and chemistry; each of them is given in two time-slots, morning and evening. The 
items (physics,morning) and (chemistry,evening) are complementary, since in the (unique) partition of the four seats into two feasible subsets, these two seats appear together. By Proposition \ref{prop:complementary}, an EF1 allocation may not exist.

In general, the above constraints can be represented by an intersection of two partition matroids: one matroid $(M, \indSets[1])$ partitions the seats into two categories by their time,  and another matroid $(M, \indSets[2])$ partitions them by their subject, and the capacities of all categories are $1$. The feasible bundles are the bundles in $\indSets[1]\cap \indSets[2]$.

The above constraints can also be represented by the set of \emph{matchings} in the following bipartite graph (where the edges are the items):
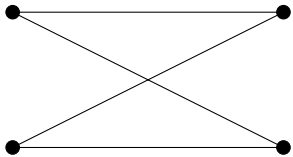
\begin{center}
\begin{tikzpicture}[scale=0.9]
\draw (0,2) -- (4,2) -- (0,4) -- (4,4) -- (0,2);
\draw (0,2)[fill] circle [radius = 0.1];
\draw (0,4)[fill] circle [radius = 0.1];
\draw (4,2)[fill] circle [radius = 0.1];
\draw (4,4)[fill] circle [radius = 0.1];
\end{tikzpicture}
\end{center}
Matching constraints in bipartite graphs can be formed as an intersection of two partition matroids: 
one partitions the edges into categories based on their leftmost endpoint, and the other partitions the edges into categories based on their rightmost endpoints, and all categories in both matroids have a capacity of $1$.
\end{example}

\begin{example}[Conflict-graph constraints]
\label{exp:conflict}
Conflict-graph constrainst were recently studied by 
\citet{hummel2021fair}. 
Suppose the items are the vertices of the graph below:
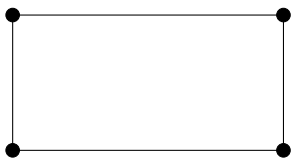
\begin{center}
\begin{tikzpicture}[scale=0.9]
\draw (0,2) -- (4,2) -- (4,4) -- (0,4) -- (0,2);
\draw (0,2)[fill] circle [radius = 0.1];
\draw (0,4)[fill] circle [radius = 0.1];
\draw (4,2)[fill] circle [radius = 0.1];
\draw (4,4)[fill] circle [radius = 0.1];
\end{tikzpicture}
\end{center}
Edges denote conflicts, and the feasible sets are the set of non-adjacent vertices. Two diagonally-opposite vertices are complementary, so by Proposition \ref{prop:complementary} an EF1 allocation may not exist.
\end{example}

\begin{example}[Budget constraints]
\label{exp:budget}
Budget constraints were recently studied by \citet{wu2021budget,gan2021approximately}.
Suppose there are two items $x,y$ with a cost of 10 and one item $z$ with a cost of 20, and two agents with budget 20. The only complete feasible partition is $(\{z\}, \{x,y\})$. The items $x,y$ are complementary, so by Proposition \ref{prop:complementary} an EF1 allocation might not exist.
\end{example}

\begin{remark}
If $(M, \indSets)$ is a matroid, and there is at least one partition of $M$ into two independent sets, then there are no complementary items for $\indSets$. This follows from the symmetric basis exchange proprety \citep{brualdi_1969}.

The opposite is not necessarily true. For example, suppose the elements of $M$ are arranged on a line and $\indSets$ contains  all the connected subsets along the line. This constraint is not a matroid, since it is not downward-closed. 
But it has no complementary items. Indeed, an EF1 allocation  exists for any number of agents with binary valuations \citep{bilo2018almost}.

As another example, consider a budget constraint with a budget of $7$, and suppose there are four items with costs $1, 2, 3, 4$. This constraint is downward-closed, but it is not a matroid, since there are maximal feasible sets of different cardinalities ($1,2,4$ and $3,4$). 
But it has no complementary items: $1$ and $2$ are separated by the feasible partition $\{1,3\},\{2,4\}$; $3$ is separated from the other items by the feasible partition $\{3\},\{1,2,4\}$; 
and $4$ is separated from the other items by the feasible partition $\{4\},\{1,2,3\}$.

Therefore, characterizing the constraints for which an EF1 allocation is guaranteed to exist remains an open problem.
\end{remark}

\subsection{Non-Existence of EFX, Uniform Matroids}
\label{sub:uniform-efx}
An \emph{Envy Free up to any good} (EFX) allocation is a feasible allocation $\Allocation{}$ where for every pair of agents $i,j$,  for every good $g$ in $j$'s bundle, $\valuation{i}{\Allocation{i}} \geq \valuation{i}{\Allocation{j}\setminus\{g\}}$. Clearly, EFX is stronger than EF1. EFX has been recently shown to exist in the unconstrained settings for up to 3 agents with additive valuations
\citep{chaudhury2020efx}. However, under constrained settings an EFX allocation may not exist even in the simple setting of two agents with identical uniform matroid constraints and identical binary valuations. 
\begin{example}
\label{exp:efx}
There are four items $a,b,c,d$, with values $v(a)=v(b)=v(c)=0$ and $v(d)=1$ for both agents, and a capacity of 2 for each agent. In every feasible allocation (i.e., allocating 2 items to each agent), the agent who does not get item $d$ is envious beyond EFX.
\end{example}

\section{At most two categories}
\label{partition_warmup}
\label{sec:2-categories}

\subsection{Uniform Matroids}
As a warm-up, we present a simple algorithm for a setting with a single category. We call it Capped Round Robin (CRR).
CRR is a slight modification of round robin, where if an agent reached her capacity --- she is skipped over (Algorithm~\ref{alg:crr}).
\begin{algorithm}[h]
\caption{Capped Round Robin
\label{alg:crr}
}
\begin{algorithmic}[1]
\REQUIRE{Category $\Category{h}$ with capacities $\capacity{i}{h}$ for every $i\in[n]$,  and an order $\sigma$ over $[n]$.}
\STATE \textbf{Initialize}: 
	$L\leftarrow \Category{h}$, $P\leftarrow \{i: \capacity{i}{h} = 0\}$, $t\leftarrow 0$, $\forall i\in [n] \  \allocation{i}{h}\leftarrow \emptyset$.
\WHILE{$L\neq \emptyset$}
	\STATE		$i\leftarrow \sigma[t]$.
	\IF{$i\notin P$}
		\STATE $g = \argmax_{g\in L}\valuation{i}{\{g\}}$.
		\STATE $ \allocation{i}{h} \leftarrow  \allocation{i}{h}\cup \{g\}$.
		\COMMENT{Agent $i$ gets her best unallocated item in $\Category{h}$}
		\STATE $L \leftarrow L \setminus \{g\}$.
		\IF{$| \allocation{i}{h}|==\capacity{i}{h}$}
			\STATE $P\leftarrow P \cup \{i\}$
			\COMMENT {Agent $i$ cannot get any more items from $\Category{h}$}
		\ENDIF
	\ENDIF
	\STATE $t \leftarrow t + 1 \mod{n}$.
\ENDWHILE
\RETURN $\allocation{}{h}$
\end{algorithmic}
\end{algorithm}

CRR finds a \fefone{} allocation whenever the constraints of all agents are \emph{uniform matroids}, i.e., all items belong to a single category (but agents may have different capacities and different valuations).

\begin{theorem}
\label{thm:uniform}
With uniform-matroid constraints, CRR finds an \fefone{} allocation. Furthermore, if $\xx$ is the outcome, for every ${i,j}$ such that $i$ precedes $j$ in $\sigma$, $\vi[{\allocation{i}{}}]=\fvi[{\allocation{i}{}}]\geq \fvaluation{i}{\allocation{j}{}}$.
\end{theorem}

The proof is similar to that of standard round robin in unconstrained setting, we include it here for completeness.

\begin{proof}
Let $i,j\in N$ s.t. $i$ precedes $j$ in $\sigma$.
First, we prove  $\vi[{\allocation{i}{}}]\geq \fvaluation{i}{\allocation{j}{}}$.
Since $i$ chooses first among $i,j$,  $|\allocation{i}{}|\geq |\feasible{i}{\allocation{j}{}}|$.
If we order $\allocation{i}{}$ and $\feasible{i}{\allocation{j}{}}$ according to the order in which the items were taken, 
every item in $\allocation{i}{}$ was chosen before (and therefore worth more to agent $i$ than) the corresponding item in $\feasible{i}{\allocation{j}{}}$ (if such one exist, as $|\allocation{i}{}|\geq |\feasible{i}{\allocation{j}{}}|$). 
That is because between the two of them, $i$ chose first.
So $\vi[{\allocation{i}{}}] \geq \vi[{\feasible{i}{\allocation{j}{}}}] = \fvaluation{i}{\allocation{j}{}}$.

Now it remains to show $\xx$ is \fefone. Let $g$ be the first item chosen by agent $i$. Notice that if we remove $g$ from $\vi[{\allocation{i}{}}]$, it is equivalent to $j$ being the one choosing first among $i,j$ (when the item set does not include $g$) and therefore we can use the exact same argument to claim that $\vj[{\allocation{j}{}}] \geq \vj[{\feasible{j}{\allocation{i}{}\setminus\{g\}}}] = \fvaluation{j}{\allocation{i}{}\setminus\{g\}}$.
\end{proof}

\subsection{Two categories}
\label{sub:2-categories}
While CRR may not find an \fefone{} allocation for more than one category, we can extend it to two categories by running CRR with reverse order on the second category; see Algorithm \ref{alg:2c-crr}. 
\begin{theorem}
\label{thm:2_categories}
When all agents have partition-matroid constraints with at most two categories, the same categories but possibly different capacities, an \fefone{} allocation always exists and can be found efficiently.
\end{theorem}

\begin{algorithm}
\begin{algorithmic}[1]
\STATE $\sigma \leftarrow$ an arbitrary order over the agents.
\STATE Run Capped Round Robin with $\Category{1}, \sigma$. Let $\allocation{i}{1}$ be the outcome for each agent $i\in N$.
\STATE $\sigma'\leftarrow \text{reverse}(\sigma)$.
\STATE Run Capped Round Robin with $\Category{2}, \sigma'$. Let $\allocation{i}{2}$ be the outcome for each agent $i\in N$.
\RETURN $\Allocation{i}{1}\cup\allocation{i}{2}$ for all $i\in N$.
\end{algorithmic}
\caption{Back-and-Forth CRR
\label{alg:2c-crr}
}
\end{algorithm}

\begin{proof}
Algorithm {\em Back-and-forth CRR} (Algorithm \ref{alg:2c-crr}) runs CRR in an arbitrary order for the first category, then uses the reverse order for CRR in the second category.
After the first category, by Theorem \ref{thm:uniform}, the allocation is \fefone \ and no agent envies another agent that appears in $\sigma$ after her.
Consider two arbitrary agents $i,j$ at the end of the algorithm.
If agent $i$ f-envied agent $j$ (up to 1 good) after the first category, she appears before $j$ in $\sigma'$ and thus will not gain any more envy in the second category. If $i$ didn't f-envy $j$ after the first category, she can only gain envy up to one good in the second category. 
That is, in one of the categories she might envy up to one good, in the other she will not envy at all. We conclude that the resulting allocation is \fefone.
\end{proof}

\section{Different Capacities, Identical Valuations}
\label{sub:identical-valuations}
\label{sec:identical-valuations}
We now consider an arbitrary number of categories, allow agents to have different capacities, but assume that all agents have the same valuations; this is, in a sense, the dual setting of \citet{biswas2018fair} who consider identical capacities and different valuations.
Using CRR as a subroutine, we show  that a similar algorithm to the one used by \citet{biswas2018fair} finds an \fefone{} allocation in this setting; this follows from the fact that no cycles can be formed in the envy graph.
Using Algorithm \ref{alg:pc-crr} we prove:

\begin{theorem}
\label{thm:diff_apacities_id_vals}
For every instance with identical additive valuations and partition matroids with identical categories (but possibly different capacities), Algorithm \ref{alg:pc-crr} returns an \fefone{} allocation.
\end{theorem}
Similarly to \citet{biswas2018fair}, 
our Algorithm \ref{alg:pc-crr} iterates over the categories running a sub-routine in each.
While they run round-robin, we run CRR. The order for the sub-routine is determined by a topological sort of the envy graph.  \citet{biswas2018fair} have an extra step of de-cycling the graph, which is not needed in our case due to the following lemma.
\begin{algorithm}
\begin{algorithmic}[1]
\REQUIRE{$M,\Category{},\capacity{i}{h}$ for every $i\in[n],h\in[l]$}
\ENSURE{an allocation $\Allocation{}$ which is $\fefone$}
\STATE Initialize:	$\sigma \leftarrow$ an arbitrary order over the agents; $\forall i\in[n] \ \Allocation{i}\leftarrow \emptyset$.
\FORALL{$\Category{h}\in \Category{}$}
\STATE Run Capped Round Robin with $\Category{h}, \sigma$.
\STATE $\forall i\in[n] \ \Allocation{i}\leftarrow \Allocation{i}\cup \allocation{i}{h}$.
\STATE Set $\sigma$ to be a topological order of the feasible-envy graph (which is acyclic by Lemma \ref{no_cycles_pc_crr}).
\ENDFOR
\end{algorithmic}
\caption{Per-Category Capped Round Robin
\label{alg:pc-crr}
}
\end{algorithm}

\begin{lemma}
\label{no_cycles_pc_crr}
For any setting with identical valuations (possibly with different capacities), the feasible envy graph of any feasible allocation is acyclic.
\end{lemma}

\begin{proof}
Assume towards contradiction that there exists some allocation $\xx$ whose corresponding feasible envy graph contains a cycle; denote the agents in the cycle $1, \ldots, p$, according to their order on the cycle. 
Note that, while all agents have the same valuation $v$, their feasible-valuation functions $\fvaluation{i}{}$ may differ.
\begin{itemize}
\item Feasible envy implies that, for every $i\in[p]$,
$\fvaluation{i}{\Allocation{i}} < \fvaluation{i}{\Allocation{i+1}}$, where we denote $p+1\equiv 1$.
\item Since the allocation is feasible, 
$\fvaluation{i}{\Allocation{i}} = v({\Allocation{i}})$.
\item 
Since the feasible valuation is at most the valuation, 
$\fvaluation{i}{\Allocation{i+1}} \leq v(\Allocation{i+1})$.
\end{itemize}
Combining the above three inequalities implies $v(\Allocation{i}) < v(\Allocation{i+1})$ for all $i\in[p]$, so we have $v(\Allocation{1}) < \cdots < v(\Allocation{p}) < v(\Allocation{1})$, a contradiction.
\end{proof}

\begin{proof}[Proof of Theorem~\ref{thm:diff_apacities_id_vals}]
We show by induction that after every category the allocation is \fefone.
\emph{Base}: after the first category the allocation is \fefone{} according to Theorem \ref{thm:uniform}.
\emph{Step}: assume the allocation is \fefone{} after $t$ categories. 
Before running category $t+1$ we reorder the agents topologically according to the feasible envy graph, and use this order as $\sigma$ in Algorithm \ref{alg:crr} (CRR). This is possible by Lemma \ref{no_cycles_pc_crr}, which shows that the feasible envy graph is acyclic. 
For every $i,j$ such that $i$ precedes $j$ in $\sigma$, $j$ does not F-envy $i$. 
By Theorem \ref{thm:uniform}, during category $t+1$, $j$ can become envious of $i$, but only up to one good, and $i$'s envy cannot increase. 
This implies that if the allocation is \fefone{} in the end of category $t$, it remains \fefone{} after category $t+1$.
\end{proof}

The following theorem shows that, for identical valuations and possibly different capacities, 
the Maximum Nash Welfare allocation is F-EF1.

\begin{theorem}
\label{mnw_identical_vals}
For different capacities and identical valuations, any feasible allocation that maximizes Nash Social Welfare is \fefone{}.
\end{theorem}

\begin{proof}
Let $\xx$ be an allocation that maximizes Nash social welfare (MNW), and suppose there exist agents $i,j$ such that $i$ f-envies $j$.
This means that, for at least one category $h$,
$\fvaluation{i}{\allocation{i}{h}}<\fvaluation{i}{\allocation{j}{h}}$. 
Since the allocation is feasible, this implies 
$v(\allocation{i}{h})<\fvaluation{i}{\allocation{j}{h}}$.

Without loss of generality, we can assume that $|X_i^h| = k_i^h$; otherwise, we can add to 
category $h$ some $k_i^h-|X_i^h|$ dummy elements with value $0$ to all agents and give them to agent $i$ without affecting the valuations.

$v(\allocation{i}{h})<\fvaluation{i}{\allocation{j}{h}}$ implies $\fvaluation{i}{\allocation{j}{h}}>0$, which implies  $|\allocation{i}{h}|=\capacity{i}{h}>0$
and $|\allocation{j}{h}|>0$.
Therefore, the following items exist:
\begin{align*}
b:=\argmin_{t\in \allocation{i}{h}}\vv[t];
&&
g:=\argmax_{t\in \allocation{j}{h}}\vv[t].
\end{align*}
So
\begin{align*}
\vv[{\allocation{i}{h}}]\geq \capacity{i}{h}\cdot \vv[b]; && \vv[{\allocation{j}{h}}]\leq \capacity{j}{h}\cdot \vv[g].
\end{align*}
Since $i$ f-envies $j$ in category $\Category{h}$, $\capacity{i}{h}\cdot \vv[g] \geq \fvaluation{i}{\allocation{j}{h}} > \fvaluation{i}{\allocation{i}{h}} = \vv[{\allocation{i}{h}}]\geq \capacity{i}{h}\cdot \vv[b]$, so $\vv[g]-\vv[b]> 0$.

Let $\xx'$ be the allocation obtained by $\xx$ by swapping goods $b$ and $g$ between $i$ and $j$'s allocations.
$\xx'$ is feasible since $b$ and $g$ are in the same category. 
Since $\xx$ is MNW, it follows that
\[
(\vv[{\Allocation{i}}]+\vv[g]-\vv[b])\cdot(\vv[{\Allocation{j}}]-\vv[g]+\vv[b])\leq \vv[{\Allocation{i}}]\cdot\vv[{\Allocation{j}}].
\]
Let $z=\vv[g]-\vv[b] > 0$. We get:
\[
\vv[{\Allocation{i}}]\vv[{\Allocation{j}}]-\vv[{\Allocation{i}}]z + \vv[{\Allocation{j}}]z -z^2 \leq \vv[{\Allocation{i}}]\vv[{\Allocation{j}}].
\]
Simplifying the above expression and using the fact that $z>0$, we get:
\[
\vv[{\Allocation{j}}]-z\leq \vv[{\Allocation{i}}].
\]
Since the valuation is additive, and $\vv[b]\geq 0$, we get:
\[
\vv[{\Allocation{i}}] \geq \vv[{\Allocation{j}}]-z = \vv[{\Allocation{j}}] -\vv[g] + \vv[b]
\geq  \vv[{\Allocation{j}}] -\vv[g] = \vv[{\Allocation{j}\setminus\{g\}}].
\]
By the fact that $\fvaluation{i}{S}\leq\vv[S]$ for every set $S$, 
we get:
\[
\fvaluation{i}{\Allocation{j}\setminus\{g\}} \leq \vv[{\Allocation{j}\setminus\{g\}}] \leq \vv[{\Allocation{i}}] = \fvaluation{i}{\Allocation{i}}.
\]
This implies that $\xx$ is F-EF1, completing the proof.
\end{proof}

The fact that the constraints are based on partition matroids is used in the proof step regarding the exchange of items $b$ and $g$.
Possibly, the result can be extended to agents with different matroid constraints
$(M,\indSets[i])$, as long as the different matroids satisfy some ``pairwise basis exchange'' property. Though, it is not clear how to define such a property.

\section{Partition Matroids with Binary Valuations}
\label{partition-binary}

In this section we assume that all agents have \emph{binary} additive valuations.
For this setting, we
present an efficient algorithm  that finds an \fefone{} allocation for $n$ agents with different valuations, and partition matroids with different capacity constraints.
For binary valuations $v_i(j)\in \{0,1\}$ for all $i,j$, and for every agent $i$ we refer to the set of items $J_i=\{j \in M \ s.t\  \vi[j]=1\}$ as agent $i$'s {\em desired set}.

\begin{theorem}
\label{thm:partition_binary}
In every setting with partition matroids with binary valuations (possibly heterogeneous capacities and heterogeneous valuations), an \fefone{} allocation exists and can be computed efficiently by the {\em Iterated Priority Matching Algorithm} (Algorithm~\ref{alg:partition-binary}). 
\end{theorem}

\begin{algorithm}[h]
\begin{algorithmic}[1]
\STATE
\textbf{Initialize}: 
	$\forall i\in [n] \  \allocation{i}{}\leftarrow \emptyset$.
\FOR{each category $h$}
\STATE $\forall i\in [n] \  \allocation{i}{h}\leftarrow \emptyset$.
\STATE $T^h := \max_{i\in N}\capacity{i}{h}$.
\FOR{$t = 1,\ldots,T^h$}
\STATE Construct the agent-item graph $G^h_t$ (see Definition \ref{def:agent-item-graph}).
\STATE Construct the feasible-envy graph corresponding to $\xx$.
\STATE $\sigma\leftarrow$ a topological order on the feasible envy-graph.
\STATE Find a priority matching in $G^h_t$ according to $\sigma$
 (see Definition \ref{def:priority}).
\STATE For every agent $i$ who is matched to an item $g_i$: $ \allocation{i}{h}\leftarrow \allocation{i}{h}\cup\{g\}$.
\ENDFOR
\STATE	Allocate the unmatched items of $\Category{h}$ arbitrarily to agents with remaining capacity.
\STATE	$\forall i\in [n] \  \allocation{i}{}\leftarrow \allocation{i}{}\cup\allocation{i}{h}$.
\ENDFOR
\end{algorithmic}
\caption{
Iterated Priority Matching
\label{alg:partition-binary}}
\end{algorithm}

Two key tools we use are the \emph{agent-item graph} and the \emph{priority matching}, defined next.
\begin{definition}[Agent-item graph]
\label{def:agent-item-graph}
Given a category $h$ and a partial allocation $\xx$,
the \emph{agent-item graph} is a bipartite graph $G^h$, where one side consists of the agents with remaining capacity (i.e., agents such that $|\allocation{i}{h}| < \capacity{i}{h}$), and one side consists of the unallocated items of $\Category{h}$.
An edge $(i,j)$ exists in $G^h$ iff $j$ is a desired item of $i$, that is,  $\valuation{i}{j}=1$).
\end{definition}

\begin{definition}[Priority matching]
\label{def:priority}
Given a graph $G = (V, E)$, a matching in $G$ is a subset of edges $\mu\subseteq E$ such that each vertex $u\in V$ is adjacent to at most one edge in $\mu$. Given a linear order on the vertices, $\sigma [ 1 ], \dots ,\sigma [ n ]$, every matching is associated with a binary vector of size $n$, where element $i$ equals $1$ whenever vertex $\sigma [ i ]$ is matched. 
The \emph{priority matching of $\sigma$} is the matching associated with the maximum such vector in the lexicographic order. Note that every ordering $\sigma$ over the vertices  yields a potentially different priority-matching.
\end{definition}

Priority matching was introduced by \citet{roth2005pairwise} in the context of kidney exchange, where they prove that every priority matching is also a \emph{maximum-cardinality matching}; that is, it maximizes the total number of saturated vertices in $V$.\footnote{\citet{okumura2014priority} extends this result to priority classes of arbitrary sizes, and shows a polynomial time algorithm for finding a priority matching. Simpler algorithms were presented by \citet{turner2015maximium,turner2015faster}.}

The Iterated Priority Matching algorithm (Algorithm \ref{alg:partition-binary})
works category-by-category.
For each category $h$, the items of $\Category{h}$ are allocated in two phases, namely the {\em matching phase} and the {\em leftover phase}. The matching phase proceeds in several iterations, where in each iteration, every agent receives at most one item. The number of iterations is at most the maximum capacity of an agent in $\Category{h}$, denoted by $T^h := \max_{i\in N}\capacity{i}{h}$.

Given the current allocation, let $\sigma$ be a topological order over the agents in the feasible envy graph (we shall soon show that the feasible envy-graph is cycle free). 
In each iteration $t$ of the matching phase, we construct the agent-item graph $G^h_t$, and then compute a priority-matching in $G^h_t$ with respect to $\sigma$, and augment agent allocations by the obtained matches.
We then update the feasible envy graph and proceed to the next iteration, where the next set of items in $\Category{h}$ is allocated. 

After at most $T^h$ iterations, all remaining items of category $\Category{h}$ contribute value $0$ to all agents with remaining capacity, and we move to the {\em leftover phase}. In this phase, we allocate the leftover items arbitrarily among agents, respecting feasibility constraints. This is possible since a feasible allocation exists by assumption.

To prove the correctness of the algorithm, it suffices to prove that every feasible envy-graph constructed in the process is cycle-free, and that the feasible envy between any two agents is at most $1$. We prove both conditions simultaneously in the following lemma.

\begin{lemma}
In every iteration of Algorithm \ref{alg:partition-binary}:

(a) The feasible envy-graph has no cycles; 

(b) For every $i,j\in N$, $\pe{i}{j}\leq 1$.
\end{lemma}

\begin{proof}
The proof is by induction on the categories and iterations. Both claims clearly hold from the outset (i.e., under the empty allocation).
In the analysis below, we refer to states \emph{before $(h,t)$} and \emph{after $(h,t)$} to denote the states before and after iteration $t$ of category $h$, respectively. 
 
\paragraph{Proof of property (a).}
We assume that property (a) holds before $(h,1)$ (i.e., before starting to allocate items in category $h$). We prove that it holds after $(h,t)$ for every $t$.
Suppose by contradiction that after $(h,t)$ there is a cycle $i_1 \to \cdots \to i_p = i_1$  in the feasible envy-graph. By assumption (a), the cycle did not exist before category $h$, so at least one edge was created during the first $t$ steps in category $h$. Suppose w.l.o.g. that it is the edge $i_1\to i_2$.  

Let $Q_1$ be the set of items desired by $i_1$ that are allocated to $i_1$ up to iteration $t$ of category $h$, and let $q=|Q_1|$.
Agent $i_1$ must have gotten these $q$ items in the first $q$ iterations of $h$ (otherwise, there exists an iteration $\leq q$ in which $i_1$ did not get an item, but a desired item remained unallocated, contradicting maximum priority matching). 

Let $Q_2$ be the set of items desired by $i_1$ that are allocated to $i_2$ up to iteration $t$ of category $h$.
The fact that $i_1$ started to envy $i_2$ during category $h$ implies that $|Q_2|\geq q+1$ and $\capacity{i_1}{h}\geq q+1$.
Agent $i_2$ must have gotten all these items in the first $q+1$ iterations of $h$ (otherwise, one of these items could have been allocated to $i_1$ in iteration $q+1$, contradicting maximum priority matching). 
This implies that in fact $|Q_2|=q+1$.
It also implies that iteration $q+1$ is still within the matching phase, since there is an item desired by $i_1$, and $i_1$ has remaining capacity.
Therefore, $i_2$ received at least $q+1$ items within the matching phase, implying that $i_2$'s value increased by at least $q+1$ up to iteration $t$ of category $h$. 

Let $Q_3$ be the set of items desired by $i_2$ that are allocated to $i_3$ up to iteration $t$ of category $h$. 
By assumption of the envy-cycle, $i_2$ envies $i_3$ after $(h,t)$.
By the induction assumption,
$\pe{i_2}{i_3} \leq 1$
before $(h,1)$.
Since $i_2$'s value increased by at least $q+1$ up to iteration $t$ of category $h$, it must hold that $|Q_3|\geq q+1$.
We now claim that before $(h,q+1)$, at most one item of $Q_3$ was available, and $i_3$ got it in this iteration. Otherwise, one could allocate one of those items to $i_2$, and allocate the item that $i_2$ received in iteration $q+1$ (that is desired by $i_1$), to $i_1$, increasing the priority matching. 

We conclude that $i_3$ got an item at each one of the first $q+1$ iterations of category $h$, as $|Q_3|\geq q+1$.  
Since all of these iterations are within the matching phase, all of these items are desired by $i_3$. 
Therefore, $i_3$'s value increases by at least $q+1$.
Repeating this argument, we conclude that every agent along the cycle received at least $q+1$ desired items during the first $t$ steps of $h$, including agent $i_p = i_1$; but this is in contradiction to the fact that $i_1$ received $q=|Q_1|$ items.

\paragraph{Proof of property (b).}
We assume that property (b) holds for every iteration before $(h,t)$ and prove that it holds after $(h,t)$.
By the induction assumption, 
before $(h,1),\ldots, (h,t)$ we had $\pe{i_1}{i_2}\leq 1$. 
We consider several cases.

\emph{Case (1)}: before $(h,t)$, we had $\pe{i_1}{i_2} = 0$.
Since at most one item is allocated to $i_2$ at iteration $t$, we must have $\pe{i_1}{i_2} \leq 1$ after $(h,t)$.

\emph{Case (2)}: before $(h,t)$, the capacity of agent $i_1$ was exhausted, so we have 
$\valuation{i_1}{\allocation{i_1}{h}}=\capacity{i_1}{h}\geq\fvaluation{i_1}{\allocation{i_2}{h}}$.
But before $(h,1)$ we had $\pe{i_1}{i_2}\leq 1$,
and the envy cannot increase after adding $\capacity{i_1}{h}$ to  $\valuation{i_1}{X_1}$ and at most $\capacity{i_1}{h}$ 
to $\fvaluation{i_1}{X_2}$.

\emph{Case (3)}: Agent $i_1$ does not desire any item remaining before $(h,t)$. Then clearly the envy of $i_1$ cannot change during $(h,t)$.

The remaining case is that ,before $(h,t)$ we had $\pe{i_1}{i_2} = 1$,
the capacity of $i_1$ was not exhausted, and $i_1$ desires at least one item. 
Then, $i_1$ precedes $i_2$ in the topological order $\sigma$ in iteration $t$, so the priority-matching on $G^h_t$ prefers to match $i_1$, than to match $i_2$ and leave $i_1$ unallocated. Therefore, the envy of $i_1$ at $i_2$ does not increase during $(h,t)$.
%
%
%
\end{proof}


\subsection{Pareto Efficiency}
\label{partition-binary-pe}
We show that if capacities are binary (that is, $\capacity{i}{h}\in\{0,1\}$ for all $i,h$), then Algorithm~\ref{alg:partition-binary} returns a Pareto efficient allocation, but this is not the case under arbitrary (non-binary) capacities. 

\begin{obs}
	In settings with partition constraints with heterogeneous {\em binary} capacities and heterogeneous binary valuations, Algorithm~\ref{alg:partition-binary} returns a Pareto efficient allocation.
\end{obs}

\begin{proof}
Under binary capacities, the algorithm runs a single priority matching in each category. As this matching is of maximum cardinality, it maximizes the social welfare within this category. From additivity, the allocation that maximizes SW within every category maximizes SW over all categories. Any welfare-maximizing allocation is  Pareto efficient.
\end{proof}


The following example shows that when capacities may be larger than 1, even when there is a single category and two agents with the same capacity, the allocation returned by Algorithm~\ref{alg:partition-binary} may not be Pareto efficient.

\begin{example}
Consider the setting and allocation depicted in the following table.
\begin{center}
\begin{tabular}{|c||c|c|}
\hline
Capacities & Alice & Bob\\
\hline\hline
$\capacity{A}{}=2$ & 1,1 & 0,1\\
$\capacity{B}{}=2$& 1,0 & 1,0\\
 \hline
\end{tabular}
\end{center}
There are two agents sharing an identical uniform matroid with capacity $2$, and four items: $(1,0)$, $(1,0)$, $(1,1)$, $(0,1)$ (recall that $(x,y)$ denotes an item that gives value $x$ to Alice and value $y$ to Bob).
	We claim that the allocation depicted in the table can be the outcome of Algorithm \ref{alg:partition-binary}, and is not Pareto efficient. 
Indeed, in the first iteration the priority matching may assign $(1,1)$ to Alice and $(0,1)$ to Bob. Then, Bob does not want any remaining item, so in the second iteration Alice gets the item $(1,0)$. Finally, in the leftover phase Bob gets the item $(1,0)$. In the obtained allocation, Alice has value $2$ and Bob has value $1$. This allocation is Pareto-dominated by the allocation giving items $(1,0),(1,0)$ to Alice and $(1,1),(0,1)$ to Bob, where Alice is indifferent and Bob is strictly better off.\qed
\end{example}


It remains open whether the setting with binary valuations and general heterogeneous capacities
always admits an allocation that is both \fefone{} and Pareto-efficient.

\section{Partition Matroids with Two Agents}
\label{partition_2}

\label{2_agents_sec}

In this section we present an algorithm for two agents with heterogeneous capacities.

\begin{theorem}
\label{thm:2-agents}
In every setting with two agents and partition matroid constraints, an \fefone{} allocation exists and can be computed efficiently by Algorithm $RR^2$ (Algorithm~\ref{alg:rr2}).
\end{theorem}

To present the algorithm we introduce some notation.  
\begin{itemize}
	\item Given an allocation $\xx$, the {\em surplus} of agent $i$ in category $h$ is 
	$$
	\surplus{i}{h}{\xx} :=  \fvaluation{i}{\allocation{i}{h}} - \fvaluation{i}{\allocation{j}{h}}.
	$$
That is, the difference between $i$'s value for her own bundle and her value for $j$'s bundle.
	\item Given agents $1,2$,  $\ell\in \{1,2\}$, valuation functions $\Valuation{},\Valuation{}'$ and category $h$,  $\crr{\Valuation{}}{\Valuation{}'}{\ell}^h$ is the allocation obtained by Capped Round Robin (Algorithm \ref{alg:crr} in \S\ref{sec:2-categories}) for  category $h$, under valuations 
	$\Valuation{1}=\Valuation{}, \Valuation{2}=\Valuation{}'$, and 
	where agent $\ell$ plays first.
	When clear in the context, we omit the superscript $h$ from $\crr{\Valuation{}}{\Valuation{}'}{\ell}^h$.
\end{itemize}

We are now ready to present Algorithm ``Round Robin Squared'' ($RR^2$) .
In $RR^2$, there are two layers of round robin (RR), one layer for choosing the next category, and one layer for choosing items within a category.
For every agent $i$, the categories are ordered by descending order of the surplus $\surplus{i}{h}{\crr{\Valuation{1}}{\Valuation{2}}{i}}$, that is, the surplus that agent $i$ can gain over the other agent by playing first in category $h$. Denote this order $\pi_i$. 

In the first iteration, agent 1 chooses the first category in $\pi_1$. Within this category, the items are allocated according to Capped Round Robin (CRR) (Algorithm \ref{alg:crr}), with agent 1 choosing first.
In the second iteration, agent 2 chooses the first category in $\pi_2$ that has not been chosen yet. Within this category, the items are allocated according to CRR, with agent 2 choosing first.
The algorithm proceeds in this way, where in every iteration, the agent who chooses the next category flips; that agent chooses the highest category in her surplus-order that has not been chosen yet, and within that category, agents are allocated according to CRR with that agent choosing first.
This proceeds until all categories are allocated.

\begin{algorithm}[h]
\begin{algorithmic}[1]
\REQUIRE{A set of items $M$, categories $\Category{1},...,\Category{l}$, capacities $\capacity{i}{h}$ for every $i=1,2,h\in[l]$;
$a\in\{1,2\}$ the first agent to choose.
}
\STATE \textbf{Initialize for all $i\in\{1,2\}$}: $\Allocation{i}\leftarrow \emptyset$; $\pi_i\leftarrow$ Categories listed by descending order of  $\surplus{i}{h}{\crr{\Valuation{1}}{\Valuation{2}}{i}}$.
\WHILE{there are unallocated categories}
\STATE $h\leftarrow$ the first category in $\pi_a$ not yet allocated.
\STATE Run CRR on category $h$. Let	$\allocation{}{h}\leftarrow \crr{\Valuation{1}}{\Valuation{2}}{a}^{h}$.
\STATE 	For all $i\in\{1,2\}$, let  $\Allocation{i}\leftarrow \Allocation{i}\cup\allocation{i}{h}$.
\STATE 	Switch $a$ to be the other agent.
\ENDWHILE
\end{algorithmic}
\caption{RR-Squared ($RR^2$)
\label{alg:rr2}
}
\end{algorithm}

The key lemma in our proof asserts that the surplus of an agent $i$ when playing first within a category $h$ is at least as large as minus the surplus of the same agent when playing second in the same category. Below, we denote by $-i$ the agent who is not $i$.

\begin{lemma}
\label{lem:surplus}
For every category $h$ and every $i\in\{1,2\}$:
\begin{align*}
\surplus{i}{h}{\crr{\Valuation{1}}{\Valuation{2}}{i}^{h}} \geq -\surplus{i}{h}{\crr{\Valuation{1}}{\Valuation{2}}{-i}^{h}}.
\end{align*}
\end{lemma}

We first show how Lemma~\ref{lem:surplus} implies Theorem~\ref{thm:2-agents}. Then we prove the lemma itself.
\begin{proof}[Proof of Theorem \ref{thm:2-agents}]
Without loss of generality, suppose agent 1 is the first to choose a category. 
By reordering, let $\Category{1},...,\Category{\ell}$ be the categories in the order they are chosen.
If $\ell$ is odd, we add a dummy empty category to make it even.
We show first that agent 1 does not F-envy agent 2. We have
\begin{align}
&\vone[{\Allocation{1}}] - \fvaluation{1}{\Allocation{2}} = \sum_{h=1,\ldots,\ell}\vone[{\allocation{1}{h}}]-\sum_{h=1,\ldots,\ell}\fvaluation{1}{\allocation{2}{h}} 
\nonumber && \text{(by addtivity)} \\
&=\sum_{h=1,\ldots,\ell}(\vone[{\allocation{1}{h}}]-\fvaluation{1}{\allocation{2}{h}})
\nonumber
\\
&= \sum_{h \ is \ odd}\surplus{1}{h}{\crr{\vone}{\vtwo}{1}} + \sum_{h \ is \ even}\surplus{1}{h}{\crr{\vone}{\vtwo}{2}} 
 && \text{(by definition of surplus)} 
 \label{2agents_eq3}
\\ 
&\geq \sum_{h \ is \ odd}\surplus{1}{h}{\crr{\vone}{\vtwo}{1}} + \sum_{h \ is \ even}-\surplus{1}{h}{\crr{\vone}{\vtwo}{1}} 
 && \text{(by Lemma \ref{lem:surplus})} 
\nonumber
\\
&=\sum_{t=1,\ldots,\frac{\ell}{2}}(\surplus{1}{2t-1}{\crr{\vone}{\vtwo}{1}} -\surplus{1}{2t}{\crr{\vone}{\vtwo}{1}}).
 \label{2agents_eq5}
\end{align}
Equation \eqref{2agents_eq3} follows from the definition of surplus, the facts that agent 1 chooses the odd categories, and the agent who chooses the category is the one to choose first within this category. 
Since agent $1$ chooses the odd categories, and does so based on highest surplus, it implies that for every $t$,
$\surplus{1}{2t-1}{\crr{\vone}{\vtwo}{1}} \geq\surplus{1}{2t}{\crr{\vone}{\vtwo}{1}}$, as category $2t$ was available when agent $1$ chose category $2t-1$. 
		Therefore, every summand in the sum of \eqref{2agents_eq5} is non-negative. Thus, the whole sum is non-negative, implying that  $\vone[{\Allocation{1}}] \geq \fvaluation{1}{\Allocation{2}}$,
		as desired.
	
	We next show that agent 2 does not F-envy agent 1 beyond F-EF1.
	As a thought experiment, consider the same setting with the first chosen category removed. 
	Following the same reasoning as above, in this setting agent 2 does not F-envy agent 1.
	But within the first category, agent 2 can only F-envy agent 1 up to 1 item. 
	That is, there exists one item in the first category such that when it is removed, it eliminates the feasible envy of the second agent within that category, and thus eliminates her  feasible envy altogether.
	We conclude that the obtained allocation is \fefone.
\end{proof}

Now all that is left is to prove Lemma~\ref{lem:surplus}.
The proof is based on Lemmas \ref{lem:equal_val_truthful} through \ref{lem:different_first_chooser}, which are stated below. 
All four lemmas consider a setting with two agents with identical additive valuations $\Valuation{1}=\Valuation{2}=\Valuation{}$ playing CRR (Algorithm~\ref{alg:crr}) on a single category.

\begin{lemma}
\label{lem:equal_val_truthful}
If one agent plays according to $\Valuation{}$, the best strategy for the other agent is to play according to $\Valuation{}$. That is, for every additive valuation $\Valuation{}'$ and every $\ell\in \{1,2\}$:
\begin{align*}
(a)\ \ 
\fvaluation{}{\crr[1]{\Valuation{}}{\Valuation{}}{\ell}}\geq \fvaluation{}{\crr[1]{\Valuation{}'}{\Valuation{}}{\ell}}
\\
(b)\ \ 
\fvaluation{}{\crr[2]{\Valuation{}}{\Valuation{}}{\ell}}\geq \fvaluation{}{\crr[2]{\Valuation{}}{\Valuation{}'}{\ell}}
\end{align*}
\end{lemma}
\begin{proof}
The two statements are obviously analogous; below we prove claim (b). 

Denote by  ``truthful play'' the play of CRR in which agent 2 plays according to $\Valuation{}$ and gets the bundle $\crr[2]{\Valuation{}}{\Valuation{}}{\ell}$;
denote by ``untruthful play'' the play of CRR in which agent 2 plays according to $\Valuation{}'$ and gets the bundle $\crr[2]{\Valuation{}}{\Valuation{}'}{\ell}$.
Order the items in each of these two bundles in descending order of $\Valuation{}$. Denote the resulting ordered vectors $\gamma$ and $\gamma'$ respectively, such that $v(\gamma_1) \geq v(\gamma_2)\geq \cdots$ and $v(\gamma'_1) \geq v(\gamma'_2)\geq \cdots$. Note that $|\gamma|=|\gamma'|$ (agent 2 gets the same number of items in both plays).
We now prove that $v(\gamma_t)\geq v(\gamma'_t)$ for all $t\leq|\gamma|$.

For every index $t\leq |\gamma|$, denote by $z_t$ the number of items held by agent 1 in round $t$ of agent 2, that is:
\begin{align*}
z_t := 
\begin{cases}
\min{(\capacity{1}{h},t-1)} & \text{~if~} \ell=2
\\
\min{(\capacity{1}{h},t)} & \text{~if~} \ell=1
\end{cases}
\end{align*}
Assume towards contradiction that there exists an index $t\leq |\gamma|$ s.t. $\valuation{}{\gamma'_t}>\valuation{}{\gamma_t}$ and let us look at the smallest such $t$ (corresponding to a highest valued item in $\gamma'$). 

In the truthful play, before agent 2 picks $\gamma_t$,
agents 1 and 2 together hold the $z_t+t-1$ highest-valued items; hence there are exactly $z_t+t-1$ items more valuable than $\gamma_t$.

In the untruthful play, agent 1 still plays by $v$ and thus still holds  at least $z_t$ of the $z_t+t-1$ highest-valued items.
While we do not know by which order agent 2 picks items, we do know that in the final allocation $\gamma'$, the first $t$ items are at least as valuable as $\gamma'_t$, which is by assumption more valuable than $\gamma_t$. Hence, there are at least $z_t+t$ items more valuable than $\gamma_t$; a contradiction.
\end{proof}

Since the sum of values of bundle 1 and bundle 2 is fixed, we get the following corollary:
\begin{lemma}
\label{lem:equal_val_different}
If one agent plays according to $\Valuation{}$, the worst case for this agent is that the other agent plays according to $\Valuation{}$ too. That is, for every $\Valuation{}'$ and $\ell\in \{1,2\}$:
\begin{align*}
(a)\ \ 
\fvaluation{}{\crr[2]{\Valuation{}}{\Valuation{}}{\ell}}\leq \fvaluation{}{\crr[2]{\Valuation{}'}{\Valuation{}}{\ell}}
\\
(b)\ \ 
\fvaluation{}{\crr[1]{\Valuation{}}{\Valuation{}}{\ell}}\leq \fvaluation{}{\crr[1]{\Valuation{}}{\Valuation{}'}{\ell}}
\end{align*}
\end{lemma}

Applying the proof of Lemma \ref{lem:equal_val_truthful}, but with respect to the case where the capacities of both agents are set to be the minimum of $\capacity{1}{h},\capacity{2}{h}$, gives the following lemma as a corollary:
\begin{lemma}
	\label{lem:feasible_val_extension}
	If one agent plays according to $\vv$, she (weakly) prefers the bundle the other agent gets when playing according to $\vv$ over the bundle the ther agent gets when playing according to $\vv'\neq\vv$. That is, for every $\ell \in\{1,2\}$:\begin{align*}
	(a)\ \ 
	\fvaluation{1}{\crr[2]{v}{v}{l}} \geq \fvaluation{1}{\crr[2]{v}{v'}{l}}
	\\
	(b)\ \ 
	\fvaluation{2}{\crr[1]{v}{v}{l}} \geq
	\fvaluation{2}{\crr[1]{v'}{v}{l}} 
	\end{align*}
\end{lemma}

We also use the following lemma.
\begin{lemma}
\label{lem:different_first_chooser}
The value of each agent for her own bundle when she plays first is at least her value for the other agent's bundle when the other agent plays first, and vice versa:
\begin{align*}
(a)\ \ 
\fvaluation{1}{\crr[1]{v}{v}{1}} \geq \fvaluation{1}{\crr[2]{v}{v}{2}}
\\
(b)\ \ 
\fvaluation{1}{\crr[1]{v}{v}{2}} \geq
\fvaluation{1}{\crr[2]{v}{v}{1}} 
\end{align*}
\end{lemma}
\begin{proof}
When both agents play using the same valuation, the only thing that differentiates agent $1$'s bundle when $1$ chooses first (respectively, second) from agent $2$'s bundle when $2$ chooses first (resp., second) is their capacities. 
If $\capacity{1}{h} \leq \capacity{2}{h}$, then $\crr[1]{v}{v}{1} \subseteq \crr[2]{v}{v}{2}$, and moreover, $\crr[1]{v}{v}{1}  = \feasible{1}{\crr[2]{v}{v}{2}}$. Therefore, (a) holds with equality. 
Otherwise, $\crr[2]{v}{v}{2} = \feasible{1}{\crr[2]{v}{v}{2}} \subseteq \crr[1]{v}{v}{1} $. Therefore, $\fvone[{\crr[1]{v}{v}{1}}] \geq \fvaluation{1}{\crr[2]{v}{v}{2}}$, establishing (a). Similar considerations apply to (b).
\end{proof}

With these lemmas in hand, we are ready to prove Lemma~\ref{lem:surplus}.
\begin{proof}[Proof of Lemma \ref{lem:surplus}]
We provide the proof for $i=1$; the other case is analogous.
The proof follows from the following four inequalities:
\begin{enumerate}
	\item $\fvaluation{1}{\crr[1]{\vone}{\vtwo}{1}} \geq \fvaluation{1}{\crr[1]{\vone}{\vone}{1}}$
	\item $\fvaluation{1}{\crr[2]{\vone}{\vtwo}{1}} \leq \fvaluation{1}{\crr[2]{\vone}{\vone}{1}}$ 
	\item $\fvaluation{1}{\crr[1]{\vone}{\vone}{1}} \geq \fvaluation{1}{\crr[2]{\vone}{\vtwo}{2}}$ 
	\item $\fvaluation{1}{\crr[2]{\vone}{\vone}{1}} \leq \fvaluation{1}{\crr[1]{\vone}{\vtwo}{2}}$ 
\end{enumerate}
We now prove the four inequalities above:
\begin{enumerate}
\item 
Follows by applying Lemma \ref{lem:equal_val_different}(b)
with $v := \vone$, $v' := \vtwo$, and $\ell=1$.

\item

Follows by applying Lemma \ref{lem:feasible_val_extension}(a) with $\vv := \Valuation{1}$, $\vv' := \Valuation{2}$, and $\ell=1$.

\item By applying Lemma \ref{lem:different_first_chooser}(a) with $\vv := \vone$, it follows that 
$\fvaluation{1}{{\crr[1]{\vone}{\vone}{1}}} \geq \fvaluation{1}{\crr[2]{\vone}{\vone}{2}}$. 
In addition, by applying Lemma~\ref{lem:feasible_val_extension}(a) with $\vv := \Valuation{1}$, $\vv' := \Valuation{2}$, and $\ell=2$, it follows that 
$\fvaluation{1}{\crr[2]{\vone}{\vone}{2}} \geq \fvaluation{1}{\crr[2]{\vone}{\vtwo}{2}}$. 
\item 
By applying Lemma \ref{lem:different_first_chooser}(b) with $v := v_1$, it follows that $\fvaluation{1}{\crr[2]{\vone}{\vone}{1}} \leq \fvaluation{1}{\crr[1]{\vone}{\vone}{2}} $. 
In addition, by applying Lemma \ref{lem:equal_val_different}(b) 
with $v := \vone$, $v' := \vtwo$, and $\ell=2$, it follows that 
$\fvaluation{1}{\crr[1]{\vone}{\vone}{2}} \leq \fvone[{\crr[1]{\vone}{\vtwo}{2}}]$.

\end{enumerate}

Combining the four inequalities above gives:
\begin{align*}
\surplus{1}{h}{\crr{\vone}{\vtwo}{1}} &= \vone[{\crr[1]{\vone}{\vtwo}{1}}]  - \fvaluation{1}{\crr[2]{\vone}{\vtwo}{1})} 
\\
&\geq^{1,2} \fvone[{\crr[1]{\vone}{\vone}{1}}] - \fvaluation{1}{\crr[2]{\vone}{\vone}{1}}
\\
&\geq^{3,4} \fvaluation{1}{\crr[2]{\vone}{\vtwo}{2}}-\fvone[{\crr[1]{\vone}{\vtwo}{2}}]  
\\
&= - \surplus{1}{h}{\crr{\vone}{\vtwo}{2}}.
\end{align*}
The assertion of the lemma follows.
\end{proof}

The notion of surplus allows us to treat each category as a single item, whose value for agent $i$ is $\surplus{i}{h}{\crr{\Valuation{1}}{\Valuation{2}}{i}^{h}}$.
The problem with extending this idea to three agents is that, for each category, there are $3!=6$ possible round-robin orders, so there are $6$ different potential ``surplus'' quantities. 

\section{Base-Orderable Matroids with up to Three Agents}
\label{general_matroids}

\label{3_agents}

In this section we consider constraints that are represented by a wide class of matroids, termed {\em base-orderable} (BO) matroids.%
\footnote{The class of base-orderable matroids was first introduced by \cite{BRUALDI1968244,brualdi_1969}, but the term base-orderable appeared only later (\cite{brualdi1971induced}).}
Recall that the \emph{bases} of a matroid are its inclusion-maximal independent sets.
In the definitions below, we 
\er{use the shorthands
$S+x := S\cup \{x\}$ 
and 
$S-x := S\setminus \{x\}$,
for any set $S$ and item $x$.
}

\begin{definition}
Given a matroid $(M,\indSets)$ and  independent sets $I,J\in \indSets$, items $x\in I $ and $y\in J$ represent a \emph{feasible swap} if both  $I - x + y$ and  $J  - y + x$ 
are in $\indSets$.
\end{definition}

\begin{definition} [\cite{BRUALDI1968244}]
\label{def:bo}
A matroid $\matroid=(M,\indSets)$ is \emph{base-orderable (BO)} if for every two bases $I,J\in \indSets$, there exists a \emph{feasible exchange bijection}, defined as a bijection $\mu: I\leftrightarrow J$ such that for any $x\in I$, both $I -  x + \mu(x) \in\indSets$ and $J -  \mu(x) + x \in\indSets$. 
\end{definition}

This class contains many interesting matroids, including partition matroids, laminar matroids (a natural generalization of partition matroids where the categories may be partitioned into sub-categories),%
\footnote{
Formally, a \emph{laminar matroid} is defined using $\ell$ possibly-overlapping sets $C_1, C_2, \dots, C_\ell \subseteq M$ such that $\bigcup_{h=1}^\ell C_h = M$. Additionally, for every pair $h, h' \in \{1, . . . , \ell\}$, only one of the following holds:
(i) $C_h \subset C_{h'}$ or, (ii) $C_{h'} \subset C_h$ or, (iii) $C_h \cap C_{h'} = \emptyset$. Each $C_h$ has a capacity $k_h$. A set $I \subseteq M$ is independent if and only if $|I \cap C_h| \leq k_h$ for each $h \in \{1, . . . , \ell\}$.
} transversal matroids,%
\footnote{
Let $G = (L, R, E)$ be a bipartite graph, and let $\indSets = \{A \subseteq L : \exists$ an injective
function $f_A : A \rightarrow R\}$. Then $(L, \indSets)$ is a matroid, and is called a transversal matroid.
}
and more. 
\citet{bonin2016infinite} conjectures that ``almost all matroids are base-orderable''.



When different agents have different matroids,
even when these are all partition matroids,
 an \fefone{} allocation may not exist (see Example \ref{exp:different_categories}). 
Therefore, we restrict attention to settings with a common matroid $\matroid$.
Before presenting our algorithm, we present two tools that are useful for any matroid --- BO or not:
\er{
finding a social-welfare-maximizing allocation with matroid constraints (Section \ref{sub:max-welfare}),
and extending a matroid such that every feasible partition is a partition into bases (Section \ref{sub:free-extension}).
}

\subsection{Finding a social-welfare-maximizing allocation}
\label{sub:max-welfare}
We initialize our algorithm with an allocation that maximizes the sum of agents' utilities. Such an allocation can be found in polynomial time for any common matroid constraints.
\begin{theorem}
\label{thm:swm}
For any constraints based on a common matroid, and any $n$ agents with additive valuations, it is possible to find in polynomial time, a complete feasible allocation that maximizes the sum of utilities.
\end{theorem}
\begin{proof}%
\footnote{
We are grateful to Chandra Chekuri for the proof idea.
}
The problem of SW maximization with \emph{submodular} valuations is NP-hard in general, and admits constant-factor approximations \citep{vondrak2008optimal,calinescu2011maximizing}. However, in the special case of \emph{additive} valuations with matroid constraints, it can be solved in polynomial time by reduction to the \emph{maximum-weight matroid intersection} problem:
 given two matroids over the same base-set, $(Z, \indSets[1])$ and $(Z, \indSets[2])$, where each element of $Z$ has a weight, find an element of $\indSets[1]\cap \indSets[2]$ with a largest total weight.

We construct the base set $Z := N\times M$, where each pair $(i,j)\in Z$ corresponds to allocating item $j\in M$ to agent $i\in N$. 
We construct two weighted matroids over $Z$, namely $\matroid[{1}]=(Z,\indSets[{1}])$ and $\matroid[{2}]=(Z,\indSets[{2}])$.
We first describe the independent sets in both matroids, then specify the weight function.

The first matroid, $\matroid[{1}]$, represents the original matroid constraints: $S\subseteq Z$ is in $\indSets[{1}]$ iff for every agent $i$, the set of items $S_i=\{j : (i,j) \in S\}$ is an independent set in the original matroid. 
We show that $\matroid[{1}]$ is a matroid: 
(1) $\emptyset\in\indSets[{1}]$. 
(2) Downward-closed: let $S\subseteq Z$ such that $S\in\indSets[1]$. Then for every $i$, $S_i\in\indSets$. Consider a subset $S'\subseteq S$. For every $i$, $S'_i\subset S_i$. Since $\matroid$ is  downward-closed, $S'_i\in \indSets$. By the definition of  $\indSets[1]$ we conclude that $S'\in \indSets[1]$. 
(3) Augmentation property: Let $S,S'\in \indSets[1]$ such that $|S'|>|S|$. Then there must be some index $i$ for which $|S'_i|>|S_i|$. Since $S'_i,S_i\in \indSets$, it follows from the augmentation property of $\matroid$ that there exists an item $a\in S'_i\setminus S_i$ such that $ S_i\cup\{a\}\in\indSets$. Then $S\cup \{(i,a)\}\in\indSets[1]$ and the augmentation property holds.

The second matroid, $\matroid[{2}]$, is a partition matroid with $m$ categories, where category $j$ corresponds to the set $\{(1,j),\ldots,(n,j)\}$, and every category has capacity 1. This essentially ensures that every item is given to at most one agent. 
One can easily verify that every subset of $Z$ that is an independent set in both matroids represents a feasible allocation. 

Next, define the weight function $w:Z\rightarrow\mathbb{N}$. 
For every $(i,j)\in Z$, let $w((i,j))=V + \vi[j]$, where $V := m\cdot \max_i \max_j v_i(j)$.
Using Edmonds' polynomial-time algorithm \citep{edmonds2003submodular}, find a maximum-weight subset $S^*\in \indSets[{1}]\cap \indSets[{2}]$.
The construction of $w$ guarantees that
$S^*$ maximizes the number of allocated items.
Subject to this, $S^*$ maximizes the total value.

Since $\feasibleAlloc\neq\emptyset$ (by assumption), maximizing the number of allocated items ensures that all items are allocated, namely that the allocation is complete. 
Within complete allocations, maximizing the total value ensures that the returned allocation maximizes social welfare.
This concludes the proof.
\end{proof}

\subsection{Ensuring a partition into bases}
\label{sub:free-extension}
To simplify the algorithms,
we pre-process the instance to ensure that, in all feasible allocations, every agent receives a \emph{basis} of $\matroid$.
Recall that the \emph{rank} of a matroid $\matroid$ is the cardinality of a basis of $\matroid$ (all bases have the same cardinality). We denote $r := \operatorname{rank}(\matroid)$.
In the pre-processing step, we add to $M$ dummy items, valued at $0$ by all agents, so that after the addition, $|M| = n\cdot r$. This guarantees that, in every feasible allocation, every bundle contains exactly $r$ items, so it is a basis.
To ensure that the dummy items do not affect the set of feasible allocations, we use the \emph{free extension},%
\footnote{
We are grateful to Kevin Long for the proof idea at {https://math.stackexchange.com/q/4300433}.
}  defined below.
\begin{definition}
Let $\matroid=(M,\indSets)$ be a matroid with rank $r$.
The \emph{free extension} of $\matroid$ is a matroid $\matroid'=(M',\indSets')$ defined as follows (where $\newitem$ is a new item):
\begin{align*}
M' &:= M +  \newitem;
\\
\indSets' & := \indSets ~~\cup~~ \{I + \newitem ~|~ I\in \indSets, |I|\leq r-1 \}.
\end{align*}
That is: all bundles that were previously feasible remain feasible; in addition, all non-maximal feasible bundles remain feasible when the new item is added to them.
\end{definition}
The properties of the free extension are summarized below.
\begin{observation}
If the free extension of $\matroid$ is $\matroid'$, with the new item $\newitem$, then:
\begin{itemize}
\item All bases of $\matroid$ are bases of $\matroid'$. 
\item  $\operatorname{rank}(\matroid) = \operatorname{rank}(\matroid') = r$.
\item Given a feasible partition of $M$ (a partition into independent sets), where some sets in the partition are not maximal (contain less than $r$ items), one can construct a feasible partition of $M'$ by adding $\newitem$ into some non-maximal set.
\item Given a feasible partition of $M'$, one can construct a feasible partition of $M$ by removing  $\newitem$ from the set  containing it.
\item $\matroid'$ is base-orderable if and only if $\matroid$ is base-orderable.
\end{itemize}
\end{observation}
The first four observations are trivial. We prove the fifth one in Appendix \ref{app:general}.

By Assumption \ref{asm:feasible}, our instance admits a feasible allocation. Since any feasible bundle is contained in a basis, the cardinality of every allocated bundle is at most $r$, so $|M|\leq n\cdot r$.
We construct a new instance by applying the free extension $n\cdot r - |M|$ times, 
getting a matroid $\matroid' = (M',\indSets')$ with $|M'| = n\cdot r$. 
We call the $n\cdot r - |M|$ new items \emph{dummy items}, and let all agents value them at $0$. 
\begin{observation}
The new instance satisfies the following properties.
\begin{itemize}
\item All bases of $\matroid$ are bases of $\matroid'$.
\item In every feasible allocation $(Y_1,\ldots,Y_n)$ in the new instance, $|Y_i|=r$ for all $i\in[n]$, so every $Y_i$ is a basis of $\matroid'$.
\item For every feasible allocation $(X_1,\ldots,X_n)$ in the original instance, there is a feasible allocation $(Y_1,\ldots,Y_n)$ in the new instance, where for all $i\in[n]$, $Y_i$ contains $X_i$ plus zero or more dummy items, so all agents' valuations to all bundles are identical.
\item For every feasible allocation $(Y_1,\ldots,Y_n)$ in the new instance, there is a feasible allocation $(X_1,\ldots,X_n)$ in the original instance, where for all $i\in[n]$, $Y_i$ contains $X_i$ plus zero or more dummy items, so all agents' valuations to all bundles are identical.
\item $\matroid'$ is base-orderable if and only if $\matroid$ is base-orderable.
\end{itemize}
\end{observation}
By the above observation, 
one can assume, without loss of generality, that there are exactly $n\cdot r$ items, and consider only allocations in which each agent receives a basis of $\matroid$.
We can now use the Iterated Swaps  scheme, presented in Algorithm \ref{alg:iterated-swaps}.

\begin{algorithm}[h]
\caption{
Iterated Swaps
\label{alg:iterated-swaps}}
\begin{algorithmic}[1]
\REQUIRE {Constraints based on a base-orderable matroid $\matroid$;
$n$ agents with additive valuations;
A set $M$ of items with $|M|=n\cdot \operatorname{rank}(\matroid)$.
}
\STATE \textbf{Initialize}: 	$\xx\leftarrow$ a complete feasible SWM allocation (by Theorem \ref{thm:swm}).
\WHILE{$\xx$ is not EF1}
\STATE 
\label{step:find-agents}
Find $i,j\in N$ such that agent $i$ envies agent $j$ by more than one item.
\STATE 
Find a feasible exchange bijection $\mu: \Allocation{i}\leftrightarrow \Allocation{j}$.
\STATE
\label{step:find-item}
Find an item $g_i\in\Allocation{i}$ 
such that $\vi(\mu(g_i)) > \vi(g_i)$.
\STATE
Swap items $g_i$ and $\mu(g_i)$.
\ENDWHILE
\end{algorithmic}
\end{algorithm}

The algorithm starts by finding a feasible social welfare maximizing (SWM) allocation $\xx$, using Theorem \ref{thm:swm}.
As long as there exist agents $i,j$ that violate the EF1 condition, 
the algorithm swaps a pair of items between $i$ and $j$, such that the utility of $i$ (the envious agent) increases,
and the utility of $j$ decreases by the same amount, so the allocation remains SWM.
The process terminates with an EF1 and SWM allocation.

In the next subsections we will present two settings in which Iterated Swaps is indeed guaranteed to terminate in polynomial time with an EF1 allocation.

\subsection{Three agents with binary valuations}
\label{sub:matroid-3agents-binary}
Below, we show that Iterated Swaps finds an EF1 allocation for $n=3$ agents with heterogeneous \emph{binary} valuations.

\begin{theorem}
\label{thm:3_agents}
For identical base-orderable matroid constraints, for three agents with heterogeneous binary valuations, the Iterated Swaps algorithm (Algorithm \ref{alg:iterated-swaps}) 
finds an EF1 allocation
in polynomial time.
Moreover, the resulting allocation is also social welfare maximizing, hence Pareto-efficient.
\end{theorem}

To prove Theorem~\ref{thm:3_agents} we use the following lemma.

\begin{lemma}
\label{combinedSwap}
Consider a setting with binary valuations.
Let $\xx$ be a SWM allocation in which 
agent $i$ envies agent $j$.
Let $\mu: X_i \leftrightarrow X_j$ be a feasible-exchange bijection.
Then 
there is an item $g_i\in \Allocation{i}$ for which
$\vi[g_i]=\vj[g_i]=0$
and
$\vi[g_j]=\vj[g_j]=1$, where $g_j = \mu(g_i)$.
\end{lemma}
\begin{proof}
By additivity,
\begin{align*}
v_i(X_j) &= \sum_{g\in X_j} v_i(g)
= \sum_{g\in X_i} v_i(\mu(g));
\\
v_i(X_i) &= \sum_{g\in X_i} v_i(g).
\end{align*}
Since $i$ envies $j$, $v_i(X_j)>v_i(X_i)$, so at least one term in the first sum must be larger than the corresponding term in the second sum. So there is some $g_i \in X_i$ for which:
\begin{align*}
v_i(\mu(g_i)) > v_i(g_i).
\end{align*}
Let $g_j := \mu(g_i)$.
Since valuations are binary, $v_i(g_j)>v_i(g_i)$ implies $v_i(g_j)=1$ and $v_i(g_i)=0$.
Since $\mu$ is a feasible-exchange bijection, swapping $g_i$ and $g_j$ yields a feasible allocation.
Since $\xx$ maximizes the sum of utilities among all feasible allocations, the swap cannot increase the sum of utilities; therefore, we must have $v_j(g_j)=1$ and $v_j(g_i)=0$ too.
\end{proof}

We call the exchange described by Lemma \ref{combinedSwap} a \emph{smart swap}.

\begin{lemma}
\label{envyDrop}
Let $\xx$ be a SWM allocation where  $\pe{i}{j} > 1$, and let $\xx'$ be an allocation obtained from $\xx$ by a smart swap between $i$ and $j$. Then:
\begin{enumerate}
\item $\xx'$ is a SWM allocation.
\item $\pe[\xx']{i}{j} = \pe{i}{j} - 2$;
\item $\vi[{\Allocation{j}'}] \geq \vi[{\Allocation{i}'}]$;
\item $\pe[\xx']{j}{i} = 0$.
\end{enumerate}
\end{lemma}

\begin{proof}
1. The smart swap decreases the utility of $j$ by $v_j(g_j)-v_j(g_i)=1$, and increases the  utility of $i$ by $v_i(g_j)-v_i(g_i)=1$, and does not change the utilities of other agents. So the total sum of utilities does not change.

2. After the smart swap, 
 $\vi[{\Allocation{j}'}] = \vi[{\Allocation{j}}]-\vi[g_j] = \vi[{\Allocation{j}}]-1$, 
 and
 $\vi[{\Allocation{i}'}] = \vi[{\Allocation{i}}]+\vi[g_j] = \vi[{\Allocation{i}}]+1$.
It follows that
$\pe[\xx']{i}{j}=\vi[{\Allocation{j}'}] -\vi[{\Allocation{i}'}] 
= \vi[{\Allocation{j}}]-1-(\vi[{\Allocation{i}}]+1) = \pe{i}{j}-2$.
 
3. Before the swap, we had $\pe{i}{j}\geq 2$, so
$\vi[{\Allocation{j}'}] \geq \vi[{\Allocation{i}'}]+2$.
The swap decreased the difference in utilities by 2. Therefore, after the swap, we still have $\vi[{\Allocation{j}'}] \geq \vi[{\Allocation{i}'}]$.

4.
If we had $\pe[\xx']{j}{i} > 0$, then giving $\Allocation{i}'$ to $j$
and $\Allocation{j}'$ to $i$
 would increase the utility of $j$ and not decrease the utility of $i$ (by 3), contradicting SWM.
\end{proof}

We are now ready to prove Theorem~\ref{thm:3_agents}. In the proof, when we mention a change in $\pe{i}{j}$ we refer to the change in the positive envy of agent $i$ to agent $j$ between allocations $\xx$ and $\xx'$; i.e., to $\pe[\xx']{i}{j} - \pe{i}{j}$.

\begin{proof}[Proof of Theorem \ref{thm:3_agents}] 
Let $\xx$ be a complete feasible SWM allocation. 
Let $\Phi(\cdot)$ be the following potential function:
	\[
	\Phi(\xx) := \sum_{i}\sum_{j\neq i} \pe{i}{j}.
	\]
If $\xx$ is EF1, we are done. 
Otherwise, 
by Lemmas \ref{combinedSwap} and \ref{envyDrop}, there must exist a smart swap between $i,j$ such that the social welfare remains unchanged, $\pe{i}{j}$ drops by 2 and $\pe{j}{i}$ remains 0. Thus, $\pe{i}{j} + \pe{j}{i}$ drops by $2$.
Let us next consider the positive envy that might be added due to terms of $\Phi$ that include the third agent, deonte it by $k$.
\begin{enumerate}
\item $\pe{i}{k}$ cannot increase, as the smart swap increases $i$'s utility, while $\valuation{i}{\Allocation{k}}$ does not change.
\item $\pe{k}{i}$ increases by at most $1$: the largest possible increase in $\valuation{k}{\Allocation{i}'}$ is $1$, while $\valuation{k}{\Allocation{k}}$ does not change.
\item $\pe{k}{j}$ increases by at most $1$: the largest possible increase in $\valuation{k}{\Allocation{j}'}$ is $1$, while $\valuation{k}{\Allocation{k}}$ does not change.
\item $\pe{j}{k}$ increases by at most $1$, as this is the exact decrease in $\vj[{\Allocation{j}'}]$, while $\vj[{\Allocation{k}}]$ does not change.
\end{enumerate}
We next claim that among the terms that may increase by $1$ ($\#2,\#3,\#4$), no two of them can increase simultaneously: 
\begin{itemize}
\item $\pe{k}{j},\pe{j}{k}$ cannot increase simultaneously as this would create an envy-cycle, contradicting SWM.
\item $\pe{k}{i},\pe{j}{k}$ cannot increase simultaneously, as this together with the fact that $\vi[{\Allocation{j}'}] \geq \vi[{\Allocation{i}'}]$ contradicts SWM: shifting bundles along the cycle $i \rightarrow j \rightarrow k \rightarrow i$ strictly increases the sum of utilities.
\item $\pe{k}{i},\pe{k}{j}$ cannot increase simultaneously as the sum of $k$'s values to $i$ and $j$'s bundles is fixed, that is,  $\valuation{k}{\Allocation{i}} + \valuation{k}{\Allocation{j}}=\valuation{k}{\Allocation{i}'} + \valuation{k}{\Allocation{j}'}$. 
\end{itemize}
We conclude that in every iteration the potential function drops by at least $1$. Indeed, $\pe{i}{j}$ drops by $2$, $\pe{j}{i}$ remains 0, $\pe{i}{k}$ does not change, and among $\pe{k}{i}$, $\pe{k}{j}$, $\pe{j}{k}$ only one can increase, by at most $1$. 

As the valuations are binary, the initial value of the potential function $\Phi(\cdot)$ is bounded by $|M|=m$. At every step it drops by at least one, so the algorithm stops after at most $m$ iterations. 
\end{proof}

\subsection{Two agents with additive valuations}
\label{sub:matroid-2agents-additive}
The case of three agents with heterogeneouss additive valuations remains open.
Below, we show that for \emph{two} agents with heterogeneouss additive valuations, an EF1 allocation always exists. Suppose the agents' valuations are $v_1$ and $v_2$. Using the cut-and-choose algorithm, we can reduce the problem to the case of \emph{identical} valuations:
\begin{itemize}
\item Find an allocation that is EF1 for two agents with identical valuation $v_1$.
\item Let agent $2$ pick a favorite bundle (the allocation is envy-free for $2$).
\item Give the other bundle to agent $1$ (the allocation is EF1 for $1$).
\end{itemize}
It remains to show how to find an EF1 allocation for agents with identical valuations.
\citet{biswas2018fair}[Section 7 in the full version] presented an algorithm that finds an EF1 allocation for $n$ agents with identical valuations, with constraints based on a \emph{laminar} matroid --- a special case of a BO matroid.
In fact, their algorithm can be both simplified and extended to BO matroids
using our pre-processing step and the Iterated Swaps scheme.
The main idea is that, in Step \ref{step:find-item}, we have to find a pair of items with a sufficiently large value-difference. The following lemma is proved by \citet{biswas2018fair}[p.18].
\begin{lemma}
Consider a setting with identical additive valuations $v$.
Let $\xx$ be an allocation in which 
agent $i$ envies agent $j$ by more than one item.
Let $\mu: X_i \leftrightarrow X_j$ be a feasible-exchange bijection.
Then 
there is an item $g_i\in \Allocation{i}$ for which
\begin{align*}
v(g_j) - v(g_i) \geq v(\Allocation{j})/m^2,
\end{align*}
where $g_j = \mu(g_i)$.
\end{lemma}
Choosing such an item $g_i$ in Step \ref{step:find-item} ensures that, with each swap, the value of agent $j$ (the envied agent) drops by a multiplicative factor of at least $(1-1/m^2)$. 
Additionally, in Step \ref{step:find-agents}, the envious agent $i$ is chosen such that $v(\Allocation{i})$ is smallest, and the envied agent $j$ is chosen such that $v(\Allocation{j})$ is largest among the agents that agent $i$ envies by more than one item.
This guarantees that each agent can be envied a polynomial number of times, and the algorithm terminates after polynomially many iterations.

\subsection{Non base-orderable matroids}
The Iterated Swaps technique may fail for matroids that are not base-orderable, even for two agents with identical binary valuations. The ``weak link'' is Lemma \ref{combinedSwap}, as the following example shows.

\begin{example}
Consider the \emph{graphical matroid on $K_4$} --- the clique on four vertices.
Denote the vertices of $K_4$ by $1, 2, 3, 4$ and its edges by $E = \{12, 13, 14, 23, 24, 34\}$. 
The $K_4$ graphical matroid is a matroid over the ground-set $E$, whose independent sets are the trees in $K_4$.
Consider the two bases $\{12,23,34\}$ (thick) and $\{24,41,13\}$ (thin):
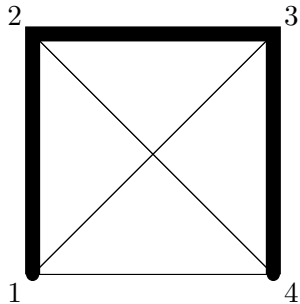
\begin{center}
\begin{tikzpicture}[scale=0.8]
\draw (0,0) -- (4,0) -- (4,4) -- (0,4) -- (0,0) -- (4,4) -- (4,0) -- (0,4);
\draw (0,0)[fill] circle [radius = 0.1];
\draw (0,4)[fill] circle [radius = 0.1];
\draw (4,0)[fill] circle [radius = 0.1];
\draw (4,4)[fill] circle [radius = 0.1];
\node at (-.3,-.3) {1};
\node at (-.3,4.3) {2};
\node at (4.3,-.3) {4};
\node at (4.3,4.3) {3};
\draw (0,0) -- (4,0) -- (4,4) -- (0,4) -- (0,0) -- (4,4) -- (4,0) -- (0,4);
\draw[line width=2mm] (0,0) -- (0,4) -- (4,4) -- (4,0);
\end{tikzpicture}
\end{center}
The only feasible swap
for $12$ is with $14$, and similarly the only feasible swap
for $34$ is with $14$, so there is no feasible-exchange bijection.

Suppose now that the agents have identical valuations, as in the following table:

\begin{table}[h]
\centering
\begin{tabular}{|c||c|c|}
\hline
Element (edge of $K_4$) & Value \\
\hline\hline
12 & 0 \\
23 & 1 \\
34 & 0 \\
\hline
13 & 1 \\
14 & 0 \\
24 & 1 \\
\hline
\end{tabular}
\end{table}
Note that, with identical valuations, all allocations are SWM.
Consider the allocation in which Alice holds the first three elements and Bob holds the last three elements. Then Alice envies Bob, but there is no swap that increases Alice's utility while keeping the bundles of both agents feasible. 
Moreover, suppose there are two copies of $K_4$, and in both copies the allocation is the same as above. Then, Alice envies Bob by two items, but there is no feasible single-item swap that can reduce her envy.
Thus, although an EF1 allocation exists, it might not be attainable by single-item swaps from an arbitrary SWM allocation. \qed
\end{example}

\section{Future Directions}
Our analysis and results suggest the following open problems.
\begin{enumerate}
\item Consider a setting with $n$ agents with additive heterogeneous valuations, partition matroids with heterogeneous capacities, and \emph{three} or more categories. Does an EF1 allocation exist? 
(\S \ref{partition_warmup} handles at most two categories).
\item Consider a setting with  $n$ agents with additive identical valuations. Is there a class of matroids, besides partition matroids with the same categories (Section  \ref{sec:identical-valuations}), for which an EF1 allocation exists even when the constraints are heterogeneous?
\item Consider a setting with  $n$ agents with 
binary valuations, partition matroids with heterogeneous capacities, and the capacities may be \emph{two} or more. Does an EF1 and Pareto-efficient allocation exist?
(\S \ref{partition-binary-pe} handles capacities in $\{0,1\}$).
\item Consider a setting with  \emph{three} or more  agents with heterogeneous additive valuations and partition matroids with heterogeneous capacities. Does an EF1 allocation exist?
(\S \ref{2_agents_sec} handles two agents).
\item Consider a setting with  \emph{four} or more agents with binary valuations and BO matroid constraints, 
or 
even 
\emph{three} agents with binary valuations and general matroid constraints,
or 
three agents with \emph{additive} valuations and BO matroid constraints. 
Does an EF1 allocation exist?
(\S\ref{general_matroids} requires three agents, binary valuations, and BO matroids).
\end{enumerate}
Another interesting direction is  extending our results to allocation of \emph{chores} (items with negative utilities) in addition to goods.

\section*{Acknowledgments}
Amitay Dror and Michal Feldman received funding from the European Research Council (ERC) under the European Union’s
Horizon 2020 research and innovation program (grant agreement No. 866132), and the Israel Science Foundation
(grant number 317/17).
Erel Segal-Halevi received funding from the Israel Science Foundation (grant number 712/20).

We are grateful to 
Jonathan Turner, 
Jan Vondrak, 
Chandra Chekuri, 
Tony Huynh, 
Yuval Filmus, 
Kevin Long, 
Siddharth Barman,
Arpita Biswas,
and anonymous reviewers of AAAI 2021
for their invaluable comments.

\newpage
\appendix
\section*{APPENDIX}

\section{Appendix for Section \ref{model_and_preliminaries}}
\label{app:model}
\label{sub:fefone}
Recall the definition of our main fairness notion from \S\ref{Fairness_Notions}:
\begin{quote}
An allocation 
$\xx{}$ is \emph{\fefone} iff for every $i,j \in N$,
there exists a subset $Y \subseteq  \Allocation{j}$ with $|Y|\leq 1$, such that $\valuation{i}{\Allocation{i}} \geq \fvaluation{i}{\Allocation{j}\setminus Y} =
\valuation{i}{\best_i(\Allocation{j}\setminus Y)}  $.
\end{quote}

This definition compares agent $i$'s bundle to $\Allocation{j}$ after first removing the most valuable item $g$ from $\Allocation{j}$, and then considering the most valuable feasible subset within $\Allocation{j} \setminus \{g\}$.

Alternatively, we could first consider the most valuable subset of $\Allocation{j}$, and then remove the most valuable item from this subset, yielding the following definition:

\begin{definition}[weakly \fefone]
An allocation 
$\xx{}$ is \emph{weakly \fefone} iff for every $i,j \in N$,
there exists a subset $Y \subseteq  \best_i(\Allocation{j})$ with $|Y|\leq 1$, such that $\fvaluation{i}{\Allocation{i}} \geq \valuation{i}{\best_i(\Allocation{j})\setminus Y}$.
\end{definition}

It is easy to see that every F-EF1 allocation is weakly F-EF1:
 if $\xx$ is \fefone, there exist $|Y|\leq 1$, $Y \subseteq  \Allocation{j}$ such that $\fvaluation{i}{\Allocation{i}} \geq \fvaluation{i}{\Allocation{j}\setminus Y}$. 
Thus,
$$
\fvaluation{i}{\Allocation{i}} \geq \fvaluation{i}{\Allocation{j}\setminus Y} = \vi[{\feasible{i}{\Allocation{j}\setminus Y}}]
=\vi[{\argmax_{T\subseteq \Allocation{j} \setminus Y, \ T\in \indSets[i]} \valuation{i}{T}}] \geq \vi[\feasible{i}{\Allocation{j}}\setminus Y],
$$
where the last inequality follows from the fact that $\feasible{i}{\Allocation{j}}\setminus Y \subseteq \Allocation{j} \setminus Y$.


However, the converse is not true; i.e., weakly \fefone{} does not imply \fefone.
Consider a uniform matroid and two agents Alice and Bob, with two items worth 1 to both agents, and let $\capacity{A}{}=1, \capacity{B}{}=2$. The allocation $\xx$ that gives both items to Bob is Weakly \fefone{} but not \fefone{}. 
Indeed, for every good $g \in \Allocation{B}$, 
$$\fvaluation{A}{\feasible{A}{\Allocation{B}}\setminus\{g\}} = \fvaluation{A}{\emptyset} \leq \fvaluation{A}{\Allocation{A}}.
$$
Therefore, $\xx$ is weakly \fefone.
On the other hand, for every $g \in \Allocation{B}$, 
$$\fvaluation{A}{\feasible{A}{\Allocation{B}\setminus\{g\}}} = 1 > \fvaluation{A}{\Allocation{A}}.
$$
Therefore, $\xx$ is not \fefone.
\begin{table}[h]
\centering
\begin{tabular}{|c||c|c|}
\hline
Capacities & Alice & Bob\\
\hline\hline
$\capacity{A}{}=1$ &  & 1,1\\
$\capacity{B}{}=2$&  & 1,1\\
 \hline
\end{tabular}
\caption{
\small
\label{tab:weak_FEF1}
An allocation that is weakly \fefone{} but not \fefone{}.
}
\end{table}

All algorithms in this paper return \fefone{} allocations, thus weakly \fefone{} as well.
Conversely, all impossibility results in this paper (see Section \ref{sec:impossibility}) continue to hold also with respect to weakly \fefone{} too: this is obvious in \S\ref{sub:partition-mnw}, 
\S\ref{sub:matching-ef1}
and 
\S\ref{sub:uniform-efx}
since they use identical constraints, for which all three fairness notions coincide.
In \S\ref{sub:different_categories} (which uses different partition-matroid constraints), it is easy to verify that the unique feasible allocation is not weakly \fefone{}.

\section{Appendix for Section \ref{general_matroids}}
\label{app:general}

\begin{lemma}
Let $\matroid = (M,\indSets)$ be a matroid, 
$\matroid' = (M',\indSets')$ be its free extension with new item $\newitem$.
Then  $\matroid'$ is base-orderable if and only if $\matroid$ is base-orderable.
\end{lemma}
\begin{proof}
If $\matroid'$ is base-orderable, then every two bases of $\matroid'$ have a feasible-exchange bijection.
Since all bases of $\matroid$ are bases of $\matroid'$, the same holds for $\matroid$ too.

Conversely, suppose $\matroid$ is base-orderable,
and let $I', J'\in \indSets'$ be two bases of $\matroid'$. We consider several cases.

\emph{Case 1:} Both $I'$ and $J'$ do not contain $\newitem$. Then both are bases of $\matroid$, so they have a feasible-exchange bijection.

\emph{Case 2:} $I'$ contains $\newitem$ while $J'$ does not. So $I' = I +  \newitem$, where $I\in \indSets$,
and $J'$ is a basis of $\matroid$. 
Let $I + y$ be any basis of $\matroid$ that contains $I$, where $y\in M$.
Since $\matroid$ is BO, there is a feasible-exchange bijection $\mu: I + y \leftrightarrow J'$.
Define a bijection $\mu': I + \newitem \leftrightarrow J'$ by: 
\begin{align*}
\mu'(x) &= \mu(x) ~~ \text{~for~} x \in I;
\\
\mu'(\newitem) &= \mu(y).
\end{align*}
We now show that $\mu'$ is a feasible-exchange bijection.
\begin{itemize}
\item For all $x\in I$, we have
\begin{align*}
(I + \newitem) - x + \mu'(x)
=
(I - x + \mu(x)) + \newitem.
\end{align*}
Since $\mu$ is a feasible-exchange bijection, 
$(I + y) - x + \mu(x) \in \indSets$.
By downward-closedness,
$I - x + \mu(x) \in \indSets$.
By definition of the free extension,
$(I - x + \mu(x))  + \newitem \in \indSets'$.
Additionally, 
\begin{align*}
J' - \mu'(x) + x 
=
J' - \mu(x) + x.
\end{align*}
Since $\mu$ is a feasible-exchange bijection,  the latter set is in $\indSets$, which is contained in $\indSets'$.
\item For $x=\newitem$, 
we have
\begin{align*}
(I + \newitem) - \newitem + \mu'(\newitem)
=
I + \mu(y)
=
(I + y) - y + \mu(y).
\end{align*}
Since $\mu$ is a feasible-exchange bijection, the latter set is in $\indSets$, which is contained in $\indSets'$.
Additionally, 
\begin{align*}
J' - \mu'(\newitem) + \newitem
=
J' - \mu(y) + \newitem.
\end{align*}
Since $\mu$ is a feasible-exchange bijection, 
$J' - \mu(y) + y \in \indSets$.
By downward-closedness,
$J' - \mu(y) \in \indSets$.
By definition of the free extension,
$(J' - \mu(y)) + \newitem \in \indSets'$.
\end{itemize}
Therefore, a feasible-exchange bijection exists for $I'$ and $J'$.

\emph{Case 3:} $J'$ contains $\newitem$ while $I'$ is a basis of $\matroid$. This case is analogous to Case 2.

\emph{Case 4:} both $I'$ and $J'$ contain $\newitem$. Similarly to Case 2, we write $I'=I+\newitem$ and find a basis $I + y$ of $\matroid$.
Let $\mu: I + y \leftrightarrow J'$ be a feasible-exchange bijection guaranteed by Case 3 above, between the bases $I + y$  and $J'$ of $\matroid'$.
Now, a bijection $\mu': I + \newitem \leftrightarrow J'$ can be defined exactly as in Case 2.
\end{proof}

\newpage

\vskip 0.2in
\bibliographystyle{theapa}
\bibliography{ref}

\end{document}